\documentclass[aps,
                pra,
                twocolumn,
                10pt,
                superscriptaddress,
                longbibliography,
                floatfix]{revtex4-2}
\usepackage{caption,subcaption}
\usepackage{amsmath,dsfont}
\usepackage{physics}
\usepackage{amsthm}
\usepackage{amssymb}
\usepackage[ruled]{algorithm2e}
\usepackage{graphicx}
\usepackage{comment}
\usepackage{mathtools}
\usepackage{enumitem}
\usepackage{tikz}
\usepackage{url}
\usepackage{bbm}
\usepackage{quantikz}
\usepackage{tabularx}
\usepackage{multirow}
\usepackage{hyperref}
\usepackage{cleveref}

\makeatletter
\newcommand{\newreptheorem}[2]{\newtheorem*{rep@#1}{\rep@title}\newenvironment{rep#1}[1]{\def\rep@title{#2 \ref*{##1}}\begin{rep@#1}}{\end{rep@#1}}}
\makeatother

\newtheorem{theorem}{Theorem}
\newreptheorem{theorem}{Theorem}
\newtheorem{algo}[theorem]{Algorithm}
\newtheorem{lemma}[theorem]{Lemma}
\newtheorem{definition}[theorem]{Definition}

\newtheorem{corollary}[theorem]{Corollary}
\newreptheorem{corollary}{Corollary}

\newtheorem{problem}[theorem]{Problem}

\newcommand{\eq}[1]{(\ref{eq:#1})}
\newcommand{\fig}[1]{\ref{fig:#1}}
\newcommand{\TII}{\affiliation{Quantum Research Center, Technology Innovation Institute, Abu Dhabi, UAE}}

\begin{document}

\title{Compilation-informed probabilistic logical-error cancellation}

\author{Giancarlo Camilo}
\thanks{Corresponding author}
\email{giancarlo.camilo@tii.ae}
\author{Thiago O. Maciel}
\author{Allan Tosta}
\author{Abdulla Alhajri}
\author{Thais de Lima Silva}
\TII
\author{Daniel Stilck Fran\c{c}a}
\affiliation{Department of Mathematical Sciences, University of Copenhagen, Universitetsparken 5, 2100 Denmark}
\author{Leandro Aolita}
\TII

\begin{abstract}
The potential of quantum computers to outperform classical ones in  practically useful tasks remains challenging in the near term due to scaling limitations and high error rates of current quantum hardware. 
While quantum error correction (QEC) offers a clear path towards fault tolerance, overcoming the scalability issues will take time. 
Early applications will likely rely on QEC combined with quantum error mitigation (QEM). 
We introduce a QEM scheme against both compilation errors and logical-gate noise that is circuit-, QEC code-, and compiler-agnostic. The scheme builds on quasi-probability methods and uses information about the circuit's gates' compilations to attain an unbiased estimation of noiseless expectation values incurring a constant sample-complexity overhead. 
Moreover, it features maximal circuit size and code distance both independent of the target precision, in contrast to strategies based on QEC alone. 
We formulate the mitigation procedure as a linear program, demonstrate its efficacy through numerical simulations, and illustrate it for estimating the Jones polynomials of knots. 
Our method significantly reduces quantum resource requirements for high-precision estimations, offering a practical route towards fault-tolerant quantum computation with precision-independent overheads for fixed circuit size and code distance.   
\end{abstract}

\maketitle

\section*{Introduction}

Quantum computers have the potential to solve certain computational problems exponentially faster than classical computers \cite{dalzell2023quantumalgorithmssurveyapplications}. 
These quantum advantages presuppose deep and fault-tolerant circuits on a large number of qubits: current estimates are hundreds of logical qubits at $10^{-6}$ error rates for scientific applications, and thousands at $10^{-12}$ rates for industrial applications~\cite{beverland2022assessingrequirementsscalepractical}. 
However, real-world devices are inherently noisy, with current error rates well above those. 
In principle, these error rates can be reduced via quantum error-correction (QEC) techniques by using multiple physical qubits to encode each logical one \cite{gottesman2009introductionquantumerrorcorrection}. Moreover, this also comes with an increase in circuit depth, as the circuit must be compiled into the code's universal gate set. 

One of the cornerstones of quantum computing is that QEC and compilation overheads are modest: both width \cite{Aharonov2008} and depth \cite{dawson2005solovaykitaevalgorithm} are increased only by polylog factors in the number of qubits, physical noise rate, and compilation error, respectively (for certain quantum LDPC codes also logarithmic in the number of logical qubits \cite{Breuckmann_2021}). 
Yet, in practice, such overheads can cause years of delay for implementations. 
These considerations led to the development of various tools to reduce the resource requirements of quantum algorithms, including quantum error mitigation (QEM) techniques \cite{Cai_2023} and early fault-tolerant (EFT) algorithms. 
In QEM, the goal is to reduce the impact of noise on expectation values without actually correcting the errors by running noisy quantum circuits multiple times and classically post-processing the outcomes. 
In turn, EFT algorithms aim to de-quantize (\emph{e.g.}, via randomization \cite{Lin_2022, Campbell2022, Tosta_2024, Wang2024}) certain components of a quantum algorithm and reserve quantum hardware strictly for crucial sub-routines. 
Both incur overhead in statistical sample complexity due to the need to estimate random variables on quantum hardware. 
 
Given the limitations of near-term quantum hardware, early applications will likely rely on simple QEC codes with a moderate qubit overhead followed by QEM at the level of the encoded qubits \cite{zimborás2025mythsquantumcomputationfault,eisert2025mindgapsfraughtroad}. 
This approach was explored in \cite{Piveteau_2021,Lostaglio_2021}, where encoded Clifford gates were protected using QEC while $T$-gate noise was mitigated using the probabilistic error-cancellation (PEC) method \cite{Temme_2017, Mari_2021}. 
The resulting sample overhead scales exponentially with the number of $T$ gates. 
Similarly, \cite{suzuki_quantum_2022} used PEC with the inversion method to mitigate decoding and compilation errors and \cite{Koczor_2024_sparse, Koczor_2024_interpolation,akibue2024,Khitrin_2026} used a similar setting for probabilistically synthesizing continuous logical quantum operations.

\begin{figure}[ht]
    \includegraphics[width=\columnwidth]{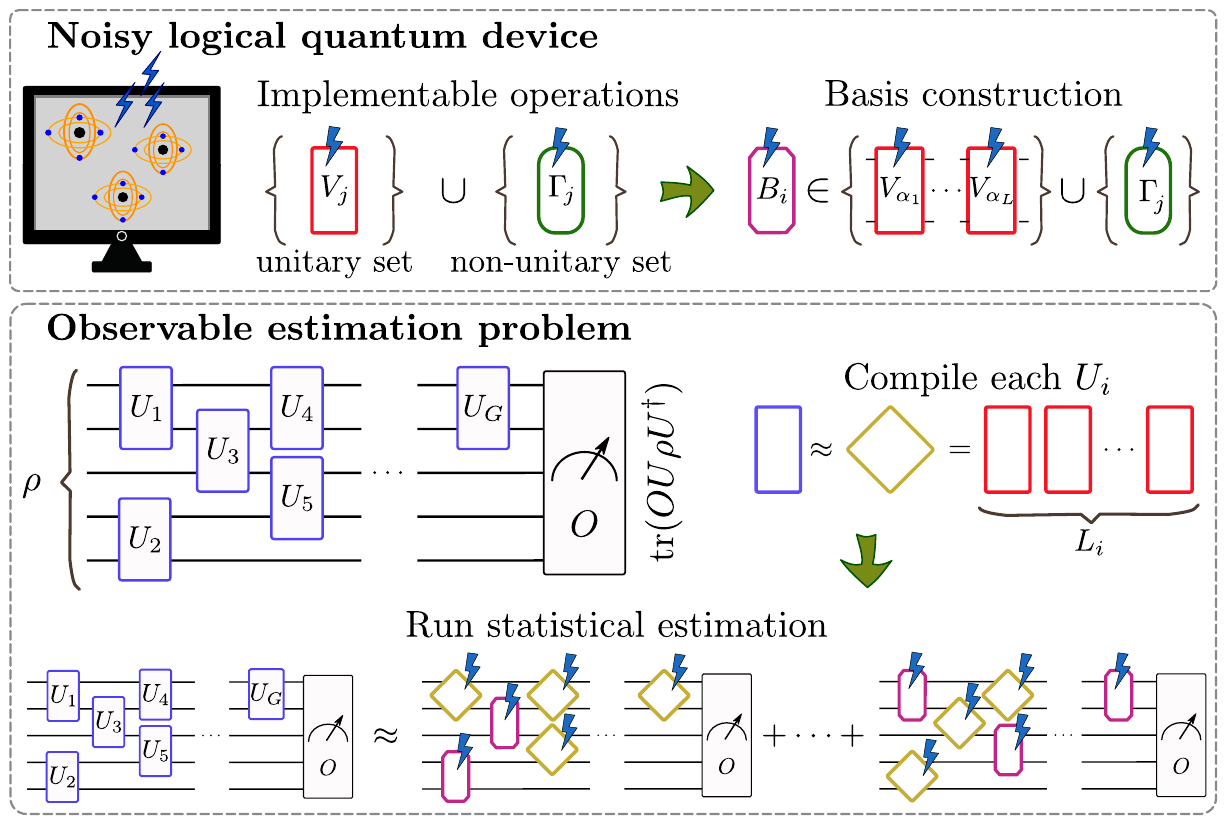}
    \caption{\textbf{Compilation-informed probabilistic error cancellation.} Upper panel: Given a characterized noisy logical quantum device, one builds a basis for the space of 2-qubit unitary channels, which includes sequences (pink boxes) of noisy logical operations (red boxes) from a universal set. 
    Lower panel: To find an $\varepsilon$-precise estimate of the expectation value $\langle O\rangle$ of an observable $O$ on an input state $\rho$ evolved by a quantum circuit $U$, each 2-qubit gate $U_i$ (blue boxes) in $U$ is first compiled up to $\varepsilon$-independent precision into a sequence (golden lozenges) of ideal unitary gates from the universal set and then decomposed as an affine combination of basis channels and the noisy version of that compilation sequence. 
    The decomposition coefficients are treated as quasi-probabilities, with the noisy compiled sequence being the dominant term. 
    $\langle O\rangle$ is then estimated statistically by Monte-Carlo sampling noisy circuits from the distribution defined by them. 
    As a result, each circuit's depth and the required QEC code-distance are both independent of $\varepsilon$. 
    }\label{fig:main}
\end{figure}

Here we propose {\it compilation-informed probabilistic error cancellation} (CIPEC), an error mitigation scheme that simultaneously removes biases from compilation errors $\epsilon_{c}$ and logical-gate errors $\epsilon_{\mathtt{Q}}$ in expectation-value estimates. 
The scheme requires the characterization of noisy logical gates up to an accuracy proportional to the target one and the inverse circuit size, but is agnostic to circuit, basis, compiler, and QEC code. 
Our scheme, illustrated in Fig. \ref{fig:main}, exploits a quasi-probability decomposition similar to the compensation method of PEC \cite{endo2018errormitigation}, where unitary gates are represented as affine combinations of gate sequences drawn from the noisy native set, but at the level of logical qubits relative to a QEC code. 
For a fixed code distance and a circuit composed of $G$ two-qubit gates, by requiring $\epsilon_{c}$ to be $\mathcal{O}(1/G)$ (independent of $\varepsilon$),  CIPEC estimates expectation values to arbitrary precision $\varepsilon$ incurring a constant sample-complexity overhead and $\varepsilon$-independent fault-tolerance overheads. 
This resource reduction can be instrumental to unlock useful applications on EFT hardware: we show that, for realistic target precisions, CIPEC allows one to solve problem instances orders of magnitude larger than those achievable using only QEC (see Fig. \ref{fig:depthmax}), as illustrated explictly with numerical experiments for Jones polynomial estimations \cite{laakkonen_less_2025}.  

Moreover, we present linear programs to find optimal quasi-probability decompositions of unitary channels in our framework efficiently and  provide explicit estimates of scaling factors for compilation errors in diamond norm (see Supplementary Note 4) 
and worst-case negativity of random 2-qubit gates with respect to three different basis of noisy channels (see Tab. \ref{tab:bases}).

Throughout the paper, we use the following notations: given a unitary  operator $U$, we denote by $\mathcal{U}$ the corresponding unitary channel, acting on a state $\rho$ as $\mathcal{U}(\rho):= U\rho\,U^\dagger$, and by $\widetilde{\mathcal{U}}$ a noisy channel approximating $\mathcal{U}$. 
More generally, $\widetilde{\Lambda}$ denotes the noisy realization of any quantum channel $\Lambda$.
Errors in quantum channels are quantified by the diamond norm error $\|\widetilde{\Lambda}-\Lambda\|_\diamond$, and the composition of $L$ channels $\Lambda_L\circ\cdots\circ\Lambda_1$ is denoted as $\prod_{i=1}^L \Lambda_i$. 
We also use $\|\boldsymbol{a}\|_1=\sum_{\alpha}|a_\alpha|$ for the $\ell_1$-norm of a vector $\boldsymbol{a}$, $[K]:=\{1,\ldots,K\}$ for the list of positive integers up to $K$, and $\vert A\vert$ for the cardinality of a set $A$.

\section*{Results}

\subsection*{Problem statement and assumptions}

Our goal is to solve the following problem:

\begin{problem}[Expectation-value estimation]\label{problem}
Given a $n$-qubit state $\rho$, an observable $O$, and the description of a unitary $U$ in terms of $G$ two-qubit unitary gates $U_i$, estimate $\langle O \rangle\coloneq\trace(O\,U\rho\,U^\dagger)$ up to precision $\varepsilon>0$ with probability at least $1-\delta$. 
\end{problem}
\noindent For that, we assume access to a device with limited QEC capabilities that can perform imperfect operations at the logical level (\emph{e.g.}, due to small QEC code distance). 
The rigorous list of assumptions is presented in the Methods section. Informally, we assume the device performs at the logical level noisy versions of two-qubit unitary gates from a finite universal gate set $V$ (determined by the particular QEC code used); we denote by $\epsilon_\mathtt{Q}$ the worst-case diamond-norm error between these noisy native gates and their ideal unitary channels. 
Moreover, the set $\widetilde{B}$ of these \emph{implementable operations}, which we refer to as the \emph{noisy basis}, is assumed to span the space of logical two-qubit unitary channels in the sense that every such channel $\mathcal{U}$ can be decomposed as a linear combination of channels from $\widetilde{B}$. 
The effectiveness of different noisy bases will be compared by a parameter $c_\star$, which we define as the \emph{worst-case negativity} over the set of two-qubit unitary channels and estimate numerically using Haar random unitaries. 

We note that the ability to express any $\mathcal{U}$ in the basis $\widetilde{B}$ is the strongest assumption: implicitly, it assumes that all the implementable operations $\widetilde{\mathcal{B}}_j\in\widetilde{B}$ are correctly characterized.  In the Methods section, we analyze the cost of logical-gate characterization and discuss regimes in which this cost is not expected to be the dominant bottleneck. 
Moreover, we also show below that, for sufficiently small $\epsilon_\mathtt{Q}$, characterization is not necessary.

\subsection*{Mitigating compilation and logical errors}

In any QEC scenario, compilation of each two-qubit gate $U_i$ in the circuit into the discrete set $V$ of logical gates realized by the QEC code is unavoidable. 
This introduces compilation errors as further sources of observable estimation errors. 
However, in the setup of quasi-probability-based error mitigation methods, such as PEC, this drawback can be turned into a useful feature: as we shall see in the next section, by using its (noisy) compilation as a dominant term in the quasi-probability decomposition of each unitary channel $\mathcal{U}_i$, one can jointly mitigate the effect of compilation and logical noise with constant sample complexity overhead. 
Moreover, this can be done with a \lq\lq rough\rq\rq~compilation whose error is independent of the target precision $\varepsilon$ aimed for the observable. 

The setup is similar to standard PEC with the so-called compensation method \cite{suzuki_quantum_2022}, which was designed to mitigate hardware noise in noisy-intermediate scale quantum (NISQ) devices and uses a direct noisy implementation of the desired gates $U_i$ on hardware (\emph{e.g.}, parametrized rotation gates) themselves as the compensation term. 
However, using a compilation of $\mathcal{U}_i$ as the compensation term to collectively mitigate hardware and compilation errors is unique to the QEC scenario. 
This term can be made dominant in the quasi-probability decomposition, hence controlling the sample overhead, by properly adjusting the tolerated compilation errors. This is in contrast with PEC for NISQ devices, where controlling the sample overhead relies heavily on physical noise strength assumptions. 
Even if one tries to naively transfer the PEC framework to the logical scenario, the required compilation error (hence the compiled circuit size) ends up scaling with the target precision $\varepsilon$ (see next Section). 
Moreover, while in NISQ devices error mitigation can be restricted to entangling gates, which are typically the noisiest operations \cite{calderon2017}, this is no longer valid in a fault-tolerant setting.
For instance, $T$ gates require gate teleportation and magic state distillation \cite{Gottesman1999}, which may result in effective noise levels comparable to those of $2$-qubit gates.

\subsection*{CIPEC algorithm}

Here we introduce \emph{Compilation-Informed Probabilistic Error Cancellation} (CIPEC), an algorithm to mitigate the joint effect of compilation and logical noise with a controlled sample complexity overhead. The algorithm is described in Alg. \ref{alg:CIPEC}. The idea is to compile each two-qubit channel $\mathcal{U}_i$ into a sequence $\mathcal{V}_{c,i}^{(L_i)}$ of $L_i$ ideal unitary channels, replace $\mathcal{V}_{c,i}^{(L_i)}$ by its noisy version $\widetilde{\mathcal{V}}_{c,i}^{(L_i)}$, and then use the basis $\widetilde{B}$ of implementable operations (see Def.~\ref{def:feasible} in the Methods) to decompose the residual channel $\mathcal{U}_i-\widetilde{\mathcal{V}}_{c,i}^{(L_i)}$ in the same spirit of quasi-probabilistic methods such as PEC. 
Since this residual channel has a small diamond norm, Lem. 9 in Supplementary Note 1 
ensures that the negativity of the quasi-probability decomposition can be controlled (see Thm.~\ref{thm:main}). 
Equivalently, the procedure can be seen as decomposing each $\mathcal{U}_i$ in the augmented basis $ \widetilde{B}\cup\{\widetilde{\mathcal{V}}_{c,i}^{(L_i)}\}$ that includes also the noisy compilation term:  
the compiled circuit is then statistically simulated by sampling operations according to their importance in the decomposition, which by construction will be dominated by the noisy compilation term $\widetilde{\mathcal{V}}_{c,i}$.  
As a result, our method mitigates errors from all gates in the circuit. 
Typically, prior to compilation,  circuits are decomposed in a way that only single-qubit rotations have to be compiled into the gate set. For the purpose of CIPEC, we consider compiling two-qubit gates, thereby allowing grouping of gates in the description of $U$ and potentially reducing the number of gates to error-mitigate and the statistical overhead \cite{hsieh2024}.

\vspace{.1cm}

\begin{algo}[CIPEC]\label{alg:CIPEC} \phantom{.}\\
\\
Input: $\{U_i\}_{i\in[G]}$, $O,\rho,\varepsilon,\delta$ as in Problem \ref{problem}; a CIPEC-capable device $\mathtt{Q}$ (Def. \ref{def:feasible}) with basis $\widetilde{B}\coloneq\big\{\widetilde{\mathcal{B}}_j\big\}_{j\in[\vert \widetilde{B}\vert]}$ (Def. \ref{def:feasible}), worst-case negativity $c_\star $ (Def. \ref{def:negativity}), and able to measure in the eigenbasis of $O$; and $\omega_1>1$. \\
\\
Output: $\overline{O}$ s.t. $\big\vert \overline{O}-\langle O\rangle\big\vert\le\varepsilon$ with prob. $1-\delta$.

\vspace{.21cm}

\begin{enumerate}[leftmargin=*]
    \item For each $i\in[G]$:
    \begin{enumerate}[leftmargin=4pt]
        \item 
        (Classical) compile $U_i$ up to diamond-norm error $\epsilon_{c,i}\le\log(\omega_1)/(2c_\star G)$ using a sequence of $L_i$ noiseless gates from the universal set $V$, and use the resulting unitary $\mathcal{V}^{(L_i)}_{c,i}$ to define the augmented noisy basis $\widetilde{B}^{i}\coloneq\widetilde{B}\,\cup\big\{\widetilde{\mathcal{V}}^{(L_{i})}_{c,i}\big\}$;
        \item
        (Classical) decompose $\mathcal{U}_i$ in the basis $\widetilde{B}^{i}$ by solving the following optimization problem 
        \begin{gather}
        \begin{aligned}\label{eq:lp}
            &\min_{\mathbf{b}_{i}\in \mathbb{R}^{\vert \widetilde{B} \vert}} \quad\gamma_i \coloneq 1+\sum_{j\in[|\widetilde{B}|]}\vert b_{i,j}\vert
            \\
            \text{s.t.} &\left\{ \mathcal{U}_i = \widetilde{\mathcal{V}}_{c,i}^{(L_{i})}+\sum_{j\in[|\widetilde{B}|]}b_{i,j}\, \widetilde{\mathcal{B}}_{j}\, \right.;
        \end{aligned}
        \end{gather}
    \end{enumerate}
    
    \item Set $\gamma\coloneq\prod_{i\in[G]}\gamma_i$ (total negativity) and $M\coloneq\frac{1}{2}\gamma^2\log(2/\delta)\,\|O\|^2\varepsilon^{-2}$ (number of samples);
    \item For $s\in[M]$:
    \begin{enumerate}[leftmargin=4pt]
        \item (Classical) sample indices $j_{i,s}\in\big[\vert \widetilde{B}^i \vert\big]$ for each $i\in[G]$ according to the probability distribution $p_{i}(j)\coloneq \gamma_{i}^{-1}\vert b_{i,j}\vert$ defined by the solution to Eq. \eq{lp} (here we include also $b_{i,\vert \widetilde{B}\vert+1}\coloneq1$). Then compute the total sign $\sigma_{s}\coloneq 
        \prod_{i\in[G]} \mathrm{sgn}(b_{i,j_{i,s}})$;
        \item (Quantum) run the noisy circuit $\widetilde{\mathcal{C}}_{s}\coloneq \prod_{i\in[G]}\widetilde{\mathcal{B}}^{i}_{j_{i,s}}$ with $\rho$ as the initial state and measure $O$ on it. Using the measurement outcome $o_{s}\in\text{spectrum}(O)$, record a sample of the random variable $x_{s}\coloneq\gamma\,\sigma_{s}\,o_{s}$;
    \end{enumerate}
\end{enumerate}
{\bf Return:} $\overline{O}\coloneq\frac{1}{M}\sum^{M}_{s=1}x_{s}$ (empirical mean of $x_\text{s}$).   
\end{algo}

The properties of the optimization problem in Eq.~\eqref{eq:lp} and the compilation strategy used to compute $\mathcal{V}_{c,i}$ in Alg.~\ref{alg:CIPEC} are described in detail in the Methods. 
We note that the $1/c_\star$ scaling of the compilation error $\epsilon_{c,i}$ in Step 1 (a) of the algorithm is actually a conservative \lq\lq worst-case\rq\rq~choice (convenient for the analytical performance guarantees of Thm.~\ref{thm:main} below). In practice, one can further reduce the compilation overhead of Alg.~\ref{alg:CIPEC} by using an adaptive scheduling of compilation errors as follows: compile to some rougher precision (say, $\epsilon_{c,i}\approx1/G$) and check if the corresponding negativity $\gamma_i$ in Step 1 (b) is under control (meaning $\gamma_i\approx1+1/G$ such that the total negativity $\gamma=\prod_{i\in[G]}\gamma_i$ can still be made independent of $G$); if not, try again with a slightly smaller error (say, $\epsilon_{c,i}/2$). Although the search over $\epsilon_{c,i}$ incurs an extra classical overhead, that should not be significant.

\subsection*{Benchmarking strategies}

Besides CIPEC, we consider two benchmarking strategies for solving Problem \ref{problem}: \emph{QEC only}, which uses the noisy logical device without any error mitigation, and \emph{standard PEC}, which refers to the logical-level adaptation of the error mitigation method of \cite{Temme_2017}. To be precise, we state explicitly below what is meant by each of these strategies. Their resource estimation, along with a comparison with CIPEC, is summarized in Table~\ref{tab:strategies}. 

In QEC only, we compile $U$ into a sequence $U_\text{c}$ of native gates and use the noisy channel $\widetilde{\mathcal{U}}_\text{c}$ (satisfying $\|{\widetilde{\mathcal{U}}_c-\mathcal{U}_c}\|_\diamond \le \epsilon_c$) to statistically compute the biased estimator $\trace(O\,\widetilde{\mathcal{U}}_\text{c}(\rho))$. 
This requires a compilation error $\epsilon_c=\mathcal{O}(\varepsilon/G)$ per two-qubit gate, which by Eq.~\eqref{eq:SK} gives a total of $L=\mathcal{O}\big(G\log^{c_2}(G/\varepsilon)\big)$ native gates. Since each has noise strength at most $\epsilon_\mathtt{Q}$, the estimation succeeds as long as $L\,\epsilon_\mathtt{Q} =\mathcal{O}(\varepsilon)$ (see Lem. \ref{lemma:maxG_ECstrategy} in Supplementary Note 2 
for the rigorous statement and proof). 

In standard PEC, we compile $U$ into a sequence $U_c$ of native gates exactly as in \emph{QEC only}, then decompose the unitary channel $\mathcal{V}$ corresponding to each native gate in $\mathcal{U}_c$ in the noisy basis $\widetilde{B}$ as $\mathcal{V}=\sum_{j\in[\vert\widetilde{B}\vert]}a_j\widetilde{B}_j$, and estimate $\trace(O\,{\mathcal{U}}_\text{c}(\rho))$ by sampling operations $\widetilde{B}_j$ with probability $p_j\coloneqq\vert a_j\vert/\|a\|_1$ similarly to CIPEC. 
Since $\widetilde{\mathcal{V}}$ (the noisy version of $\mathcal{V}$) is an implementable operation, which should be in the basis by construction, and $\|\mathcal{V}-\widetilde{\mathcal{V}}\|_\diamond\le\epsilon_\mathtt{Q}$ (see Eq.~\eqref{eq:noisychannels}), the negativity of each term is $\gamma_i\le1+c_\star\epsilon_\mathtt{Q}$ and the total negativity is $\gamma\le e^{c_\star L\epsilon_\mathtt{Q}}$ (see Lem. \ref{lemma:standardPEC} in Supplementary Note 2.  
This can be made constant as long as $c_\star L\,\epsilon_\mathtt{Q}=\mathcal{O}(1)$. 

We note that the cost of classical pre-processing for PEC is larger than for CIPEC by a factor of $L/G$: for PEC, all $L$ gates of the compiled circuit are expanded in the basis $\widetilde{B}$, while for CIPEC only the $G$ original gates before compilation are expanded. Since the terms in each quasi-probabilistic decomposition are sampled independently (see Alg.~\ref{alg:CIPEC}), the classical cost is linear in the number of mitigated gates.

\subsection*{Resource comparison}

Remarkably, for a target precision $\varepsilon$, Alg. \ref{alg:CIPEC} solves Problem \ref{problem} with $\varepsilon$-independent logical circuit sizes and physical-qubit overhead. 
This is made possible by using a noisy compilation of each gate $U_i$  with $\varepsilon$-independent compilation error in Eq. \eqref{eq:lp}, which propagates to an $\varepsilon$-independent negativity $\gamma$ as shown by Lemma \ref{lemma:negativity_ciPEC} in Supplementary Note 2. 
This is the core of our main theorem, proven in Supplementary Note 2. 

\vspace{1cm}

\begin{theorem}[CIPEC]\label{thm:main}
Let $\{U_i\}_{i\in[G]}, \varepsilon,\delta,\rho, O$ as in Problem \ref{problem},  $\mathtt{Q}$ be a CIPEC-capable device with worst-case error $\epsilon_{\mathtt{Q}}$ and worst-case negativity $c_\star $, and $\omega_1,\omega_2>1$ be constants. 
Denote by $L_i$ the length of a logical-gate sequence that compiles $U_i$ up to diamond-norm error $\epsilon_{c,i}\le\log(\omega_1)/(2c_\star G)$ into the universal set $V$ realized by $\mathtt{Q}$, and let $L \coloneq\sum_{i\in[G]}L_i$ be the total number of logical gates. 
Then, if 
\begin{align}\label{eq:CIPEC_feasibility}
    L\,\epsilon_\mathtt{Q} \le \frac{\log(\omega_2)}{2c_\star}\,,
\end{align}
Alg. \ref{alg:CIPEC} solves Problem \ref{problem} on $\mathtt{Q}$ for any target precision $\varepsilon$ with constant sample overhead $\gamma^2\le\omega_1\omega_2$, i.e., using at most $\frac{1}{2}\gamma^2\,\norm{O}^2\varepsilon^{-2}\log(2/\delta)$ samples.
\end{theorem}

\noindent The constants $\omega_1,\omega_2$ are freely chosen and their role is to parametrize, respectively, the compilation and logical error contributions to the total sample overhead $\gamma^2$ (more precisely, by upper bounding each of them). One way to adjust them, \emph{e.g.}, is as follows: suppose one has a fixed budget $S=\frac{1}{2}s^2\,\norm{O}^2\varepsilon^{-2}\log(2/\delta)$ of samples one is willing to take, where $s>1$ is a given overhead factor; one can then choose the compilation precision by setting $\omega_1<s^2$ freely, and set  $\omega_2=s^2/\omega_1$ fixing the logical error budget to  $\log(\omega_2)/(2c_\star)$. Then one solves the optimization problem in Eq.~\eqref{eq:lp} and computes the associated negativity $\gamma$ and total number of gates $L$.  Thm. \ref{thm:main} guarantees that if Eq. \eqref{eq:CIPEC_feasibility} is satisfied, then $\gamma^2\le s^2$.

\begin{table}[t]
    \renewcommand{\arraystretch}{1.3}
    \begin{tabularx}{\columnwidth}{@{\extracolsep{\fill}\hspace{0pt}}|c|ccc|c|}
      \hline
      \textbf{Method} &  $\epsilon_{c,i}$ &  $L$ & \textbf{$\gamma$} & $L\epsilon_\mathtt{Q}$\\
      \hline
      CIPEC (Alg.~\ref{alg:CIPEC}) & $\mathcal{O}\big(1/G\big)$& $\mathcal{O}\big(G\log^{c_2}(G)\big)$ & ${\mathcal{O}}\big(e^{L\epsilon_\mathtt{Q}}\big)$ & $\mathcal{O}(1)$\\
      \hline
      CIPEC (Cor.~\ref{cor:cipec_noiseless}) & $\mathcal{O}\big(1/G\big)$& $\mathcal{O}\big(G\log^{c_2}(G)\big)$ & ${\mathcal{O}}(1)$ & $\mathcal{O}(\varepsilon)$\\      
      \hline
      PEC & $\mathcal{O}\big(\varepsilon/G\big)$& $\mathcal{O}\big(G\log^{c_2}\!\big(G/\varepsilon\big)\big)$ & ${\mathcal{O}}\big(e^{L\epsilon_\mathtt{Q}}\big)$ & $\mathcal{O}(1)$\\
      \hline      
      QEC only & $\mathcal{O}\big(\varepsilon/G\big)$& $\mathcal{O}\big(G\log^{c_2}\!\big(G/\varepsilon\big)\big)$ & 1 & $\mathcal{O}(\varepsilon)$\\
      \hline
    \end{tabularx}%
  \caption{
  \textbf{Resource comparison summary. } 
Required resources for solving Problem \ref{problem} using 3 strategies: CIPEC, QEC only (no error mitigation), and standard PEC error mitigation. The first row corresponds to CIPEC as given in Alg.~\ref{alg:CIPEC}, which uses a noisy basis $\widetilde{B}$ and requires logical-gate characterization, while the second row corresponds to the special regime of Cor.~\ref{cor:cipec_noiseless} where characterization is not needed since one can use a noiseless basis $B$. $\epsilon_{c,i}$ is the required compilation error (in diamond norm) per two-qubit gate $U_i$ in the original circuit. $L=\sum_{i\in[G]}L_i$ is the total number of native gates after compilation, where each $L_i$ scales as Eq.~\eqref{eq:SK} on average. $\gamma\ge1$ is the sample prefactor such that the total number of samples is $\mathcal{O}\big(\gamma^2\,\norm{O}^2\varepsilon^{-2}\log(2/\delta)\big)$. The last column shows the required \emph{feasibility condition} on $L$ and the hardware parameter $\epsilon_\mathtt{Q}$ for rigorously guaranteeing a solution to Problem \ref{problem} with constant prefactor $\gamma$. We see that QEC has the range of feasible $L\epsilon_\mathtt{Q}$ scaled down by a factor $\mathcal{O}(\varepsilon)$ compared to CIPEC. Moreover, while PEC and CIPEC have a similar-looking condition on $L\epsilon_\mathtt{Q}$, CIPEC is advantageous compared to PEC due to the removal of the $\text{polylog}(1/\varepsilon)$ overhead in $L$.
}
    \label{tab:strategies}
\end{table}

The result of Thm~\eqref{thm:main} unlocks the possibility of solving instances of Problem \ref{problem} that cannot be solved with QEC alone or with standard PEC under the same assumptions, given a fixed budget of quantum resources. 
We see from Table \ref{tab:strategies} that, unlike the other strategies, CIPEC requires $\varepsilon$-independent compilation errors to solve Problem \ref{problem}. This results in compiled circuits with an $\varepsilon$-independent number of logical gates $L$ and circuit depth. The feasibility condition $L\,\epsilon_\mathtt{Q}\leq\log(\omega_2)/(2c_\star )=\mathcal{O}(1)$ in Thm. \ref{thm:main}, which ensures that the sample overhead of error mitigation remains constant, also does not depend on $\varepsilon$. The condition can be interpreted in two ways: $i)$ for a fixed two-qubit logical error $\epsilon_\mathtt{Q}$ (\emph{i.e.}, QEC with a fixed code distance), it gives a criterion based on the size $L$ of the compiled circuit to check whether a particular instance of Problem \ref{problem} is solvable on this device; 
$ii)$ for a fixed $L$, it sets an upper bound on $\epsilon_\mathtt{Q}$ (equivalently, a minimum code distance) for which a solution can be guaranteed for all circuits made of up to $L$ native gates. 

Notice that CIPEC's physical qubit overhead, which is determined solely by $\epsilon_\mathtt{Q}$ and the physical error rate $p$, also does not depend on the target precision $\varepsilon$. 
This is not the case for error-mitigation based on \emph{standard PEC}, which has the similar-looking feasibility condition $L(\varepsilon)\,\epsilon_\mathtt{Q}=\mathcal{O}(1)$ but with $L(\varepsilon)$ carrying a polylogarithmic $\varepsilon$-dependence. Unlike CIPEC, this dependence is reflected in both physical qubit and logical-depth overheads. 
Moreover, even for fixed $L=L(\varepsilon)$ and $\varepsilon_\mathtt{Q}$, the feasibility condition indicates that solving Problem \ref{problem} via PEC with constant sample overhead is only possible if the target $\varepsilon$ is above a threshold value. 
The same conclusion holds also for the \emph{QEC only} strategy, which also has $L=L(\varepsilon)$ and the even tighter feasibility condition $L(\varepsilon)\,\epsilon_\mathtt{Q}=\mathcal{O}(\varepsilon)$ that scales down by a factor of $\varepsilon$ the range of feasible $L$ or the maximum tolerated $\epsilon_\mathtt{Q}$. 
Fig. \ref{fig:depthmax} shows, for each strategy, the maximum logical error rate $\epsilon_\mathtt{Q}$ for which an $\varepsilon$-precise observable estimation can be guaranteed for a circuit made of $G$ gates (assuming, for the sake of illustration only, that each of them compile into the same number of native gates as a Haar random unitary). 
The results indicate a clear advantage of CIPEC in high-precision regimes. 
We notice that the maximum tolerated $\epsilon_\mathtt{Q}$ scales as $\mathcal{O}(\varepsilon/L)$ for all strategies, such that after implementing all the $L$ noisy native gates there is still enough precision remaining to deliver the $\varepsilon$-precise observable estimate. The advantage observed in Fig.~2 boils down to the fact that CIPEC has a better $L$ due to the $\varepsilon$-independence of its compilation errors $\epsilon_{c,i}$.  
Finally, we stress that CIPEC is agnostic to the choice of compiler, basis, and QEC code: any improvement in each of these modules can be easily incorporated.

We also note that, even for instances where the condition $L\,\epsilon_\mathtt{Q}\leq\log(\omega_2)/(2c_\star )$ is not satisfied, CIPEC still correctly solves Problem \ref{problem} for general $\varepsilon$, though with an exponential sample overhead $\gamma^2\le \omega_1\,e^{2c_\star L\epsilon_\mathtt{Q}}$ (see Supplementary Note 2). 
This is remarkable because there exists (see Corollary \ref{cor:unsolvable} in Supplementary Note 2) 
high-precision instances of Problem \ref{problem} that cannot be solved using QEC only or PEC even in the limit of infinite statistical samples.

\subsection*{CIPEC stability under characterization errors}

The theoretical guarantees of Thm. \ref{thm:main} refer to CIPEC (Alg.~\ref{alg:CIPEC}) with a perfect characterization of the noisy implementable operations (\emph{cf.} Eq \eqref{eq:noisychannels}). However, this is impossible in general, since standard tools such as gateset tomography (GST) can only approximately characterize them to some prescribed precision (see Methods for a discussion on the cost of GST). 
In Supplementary Note 3 (see Thm. \ref{thm:stability}) we prove that CIPEC is stable against characterization errors in the following sense: 

\begin{theorem}[Stability of CIPEC (informal)]\label{thm:stability}
CIPEC using imperfectly characterized basis channels with characterization errors $\epsilon_\text{char}\leq\varepsilon/(L\norm{O})$ estimates the observable $\langle O\rangle$ with precision $2\varepsilon$. 
\end{theorem}
\noindent Under this condition, the intended precision $\varepsilon$ is ensured in the presence of characterization errors by running CIPEC with a target precision $\varepsilon/2$. 

In fact, it follows immediately from this result (see Supplementary Note 3 for the proof) that, if $\epsilon_\mathtt{Q}\le\varepsilon/(L\|O\|)$, CIPEC  completely bypasses the need for logical noise characterization: one can characterize the noise channels simply as identity channels and perform the quasi-probability decompositions in Eq.~\eqref{eq:lp} with respect to noise-free gates and basis elements: 
\begin{corollary}[CIPEC without logical-gate characterization]\label{cor:cipec_noiseless}
If $L\,\epsilon_\mathtt{Q}\le\varepsilon/\|O\|$, CIPEC using a noiseless basis $B$ estimates the observable $\langle O\rangle$ with precision $2\varepsilon$. 
\end{corollary}
\noindent Interestingly, even in this regime, CIPEC's overheads due to error correction and to compilation are strictly smaller than those of QEC alone, since it only requires compilation to $\varepsilon$-independent precision. 
Notice, however, that in this case the mitigation is only of compilation errors, similarly to the proposal in  \cite{endo2018errormitigation} based on logical-level PEC with the error inversion method. 

\begin{figure}[t]
    \includegraphics[width=\columnwidth]{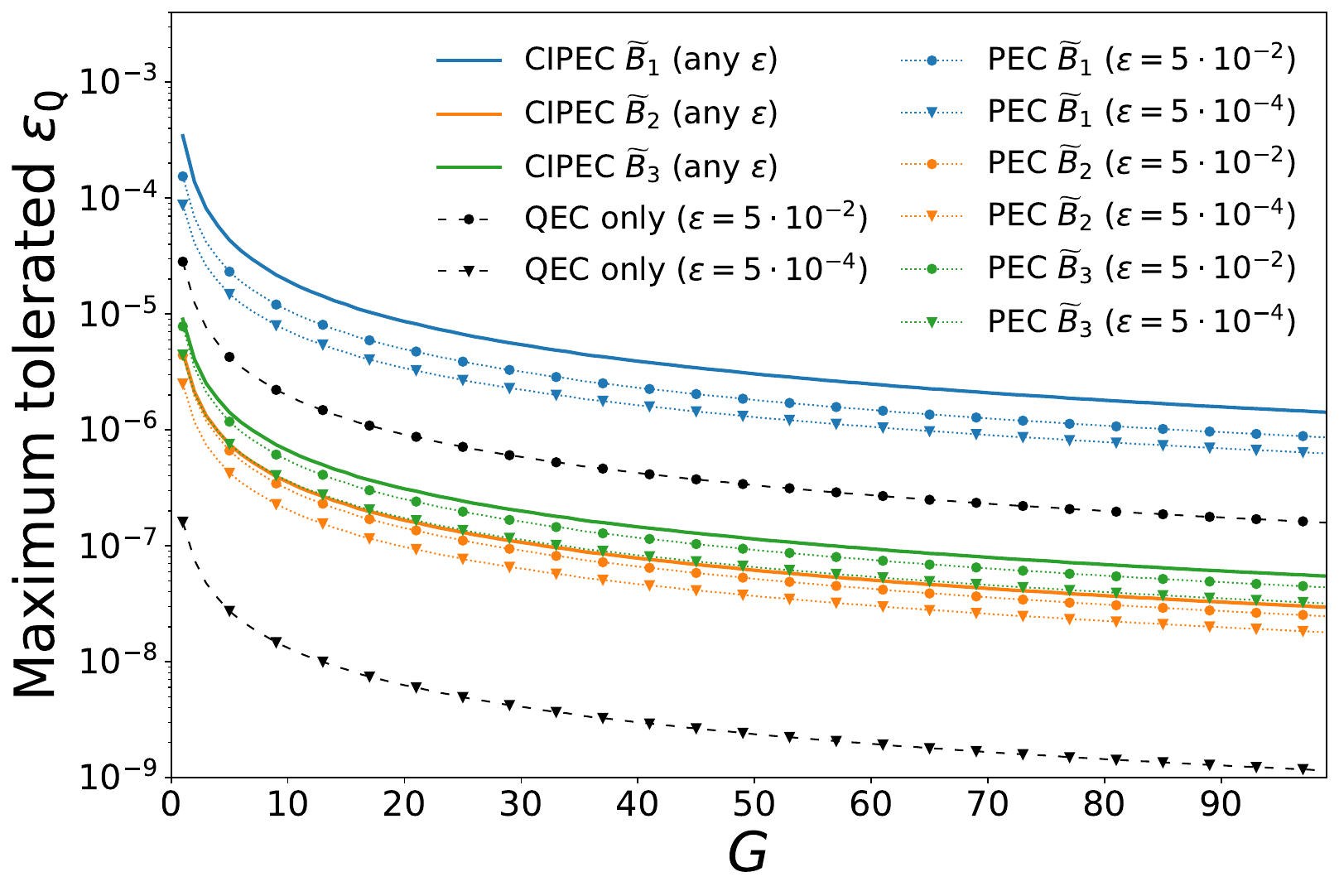}
    \caption{
    {\bf Maximum tolerated logical error $\epsilon_\mathtt{Q}$.}  
    Here we consider Problem \ref{problem} for a class of synthetic circuits $U$ made of $G$ two-qubit gates $\{U_i\}_{i\in[G]}$ before compilation. We assume that each $U_i$ is compiled to a given precision $\epsilon_{\text{c},i}$ using a number of native gates given by $L_i=210\log^{0.75}(1/\epsilon_{c,i})$, which is the average-case scaling for Haar random unitaries (see Eq.~\eqref{eq:SK} in the Methods). 
    This is for illustration purposes only, though the qualitative conclusions should remain true for any other family of circuits. 
    The compilation errors $\epsilon_{c,i}$ used for CIPEC were as prescribed by Thm.~\ref{thm:main}, and similarly by Lemmas \ref{lemma:maxG_ECstrategy} and \ref{lemma:standardPEC} in the Supplementary Information  for QEC and PEC. 
    The curves depict (in log scale) the maximum $\epsilon_\mathtt{Q}$ for which each strategy is guaranteed to solve Problem \ref{problem} with target precision $\varepsilon$. Solid lines refer to CIPEC, dotted lines to PEC, and dashed lines to QEC only. For CIPEC and PEC, we compare the three different bases of Table~\ref{tab:bases} using different colors: $\widetilde{B}^{(1)}$ (blue), $\widetilde{B}^{(2)}$ (orange), and $\widetilde{B}^{(3)}$ (green), while QEC curves have a fixed color (black). For PEC and QEC, we compare different target precisions $\varepsilon$ using markers: $\varepsilon=0.05$ (circles) and $\varepsilon=0.0005$ (triangles); CIPEC has no markers since the condition is the same for every $\varepsilon$. 
    We see that, regardless of the basis used, CIPEC accepts higher levels of logical noise than the others and the difference grows as $\varepsilon$ decreases (within an order of magnitude compared to standard PEC, and multiple orders of magnitude compared to QEC only). 
    }
    \label{fig:depthmax}
\end{figure}

\subsection*{Application: Jones polynomial estimation}

As an illustration, we apply our framework to the problem of estimating the Jones polynomial $J_K(q)$ of a knot (or link) $K$ \cite{Kauffman2001}. 
This knot invariant is a polynomial over the complex variable $q$ with integer coefficients determined solely by the knot topology. 
We recall that any knot can be represented as the closure of some braid word $\Sigma$ (a product of braid group generators) on a number $n_\text{s}$ of strands. For simplicity here we will use plat closure representation, in which case $n_\text{s}$ is even \cite{kauffman_fibonacci_2008}.  
We focus on the estimation of $J_K(q)$ at the special point $q=e^{2\pi i/5}$ to relative precision $\varepsilon$, which is known to be BQP-complete if the total number of braid crossings $\vert\Sigma\vert$ and the target precision $\varepsilon$ scale polynomially with $n_\text{s}$ (see Methods for details). 

As shown in \cite{laakkonen_less_2025}, the Jones polynomial is proportional to the matrix element $\expval{y\vert U_\Sigma\vert y}$, where $\ket{y}\coloneq\ket{0101\cdots010}$ is a $(n_\text{s}+1)$-qubit computational basis state and $U_\Sigma$ is a unitary representation of the braid word $\Sigma$. 
We use the control-free Hadamard test quantum algorithm proposed in \cite{laakkonen_less_2025} (see Methods) to statistically estimate $\expval{s\vert U_\Sigma\vert s}$ in Eq. \eqref{eq:jones}, which falls within the scope of Problem \ref{problem} with $U\rho \,U^\dagger=\ketbra{s}{s}$, $G=\mathcal{O}(\vert\Sigma\vert+ n)$ two-qubit gates, and $O=\frac{1}{2}(U_\Sigma^{\phantom{\dagger}}+U_\Sigma^{{\dagger}})$ or $O=\frac{1}{2i}(U_\Sigma^{\phantom{\dagger}}-U_\Sigma^{{\dagger}})$ for the real and imaginary parts of $\expval{s\vert U_\Sigma\vert s}$, respectively. 
Thus, for a generic knot $K$ having $n_\text{s}$ strands and $\vert\Sigma\vert$ braid crossings, CIPEC successfully estimates $J_K(e^{2\pi i/5})$ as long as $\epsilon_\mathtt{Q}\le\log(\omega_2)/(2c_\star L)$, where $L$ is the total circuit size after compilation, which by Eq.~\eqref{eq:SK} with $\epsilon_\text{c}=\mathcal{O}(1/G)$ from Thm.~\ref{thm:main} is
\begin{align}
    L=\mathcal{O}\big((\vert\Sigma\vert+n_\text{s})\log^{c_2}(\vert\Sigma\vert+n_\text{s})\big).
\end{align}

\begin{figure}[t]
    \includegraphics[width=\columnwidth]{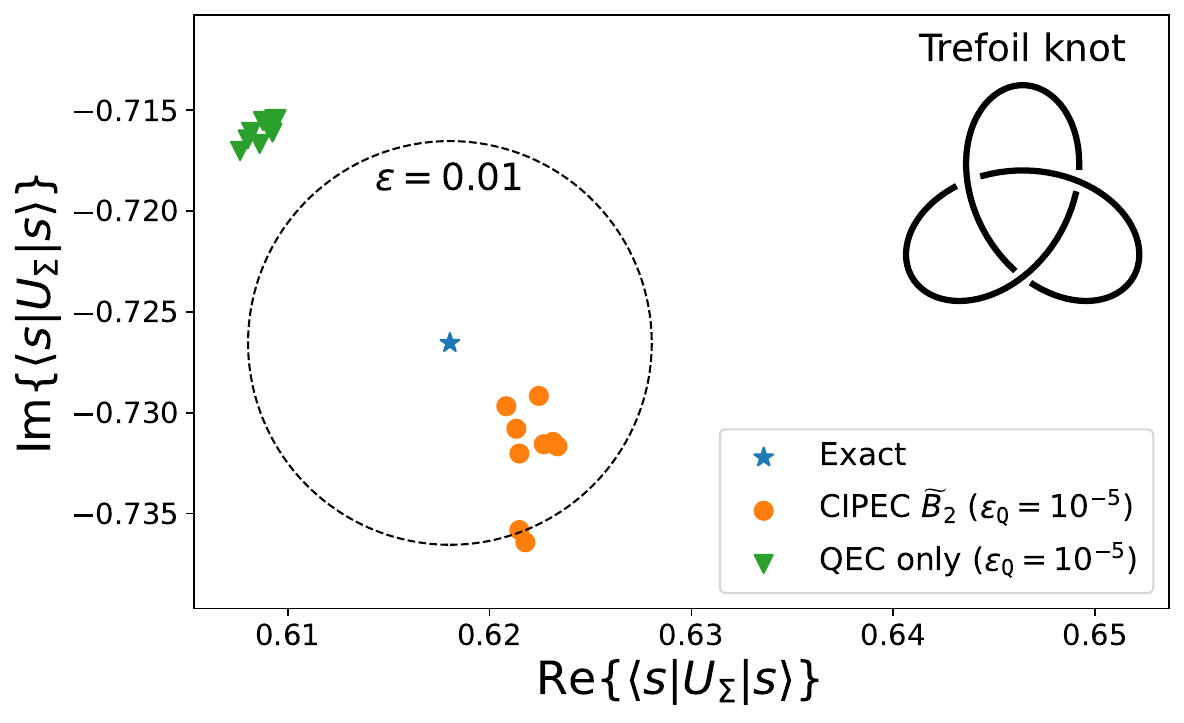}
  \caption{{\bf CIPEC for Jones polynomial estimation.}\break  
    Quantum estimation of the Jones polynomial at $q=e^{i \frac{2\pi}{5}}$ for the trefoil knot (inset) using the control-free Hadamard test of \cite{laakkonen_less_2025}. 
    The target relative precision is $\varepsilon=10^{-2}$ and we use {\sc Qibo} \cite{Qibo2021} to simulate a device of $n_{\ell}=5$ noisy logical qubits. 
    For convenience, we display the results in terms of the scaled quantity $\bra{s}U_{\Sigma}\ket{s}\propto J_{\text{trefoil}}\big(e^{i \frac{2\pi}{5}}\big)$ defined in Eq.~\eqref{eq:jones}. 
    The purple star marks the exact result $J_{\text{trefoil}}(q)=q^{-1}+q^{-3}-q^{-5}$;  the black (\lq\lq$\times$\rq\rq-shaped) points are the result of 10 distinct trial estimations without error mitigation using a shot-noise simulation of the compiled circuit for $U_{\Sigma}$; 
    the orange (\lq\lq$+$\rq\rq-shaped) points show the result of 10 distinct trials using CIPEC with the minimal basis $\widetilde{B}^{(2)}$ and a sample overhead $\gamma^2\approx2.46$; 
    the number of samples in both cases follows the expressions in Supplementary Note 2 with a failure probability $\delta=0.1$. 
    We see that CIPEC consistently delivers estimates within the required precision while the unmitigated estimates fail.
    } 
    \label{fig:jones}
\end{figure}

In Fig. \fig{jones} we show the results for the trefoil knot assuming a noisy logical device realizing Clifford + $T$ operations with local depolarizing noise after each gate. 
The corresponding Hadamard test circuit can be rearranged in terms of $G=9$ ideal two-qubit gates $U_i$ (see Methods), which were then compiled into the native Clifford + $T$ gateset using the {\sc gridsynth} algorithm \cite{ross2016}. 
The depolarizing strengths were chosen as $10^{-6}$ for single-qubit gates and state preparation channels, and $10^{-5}$ for CNOT and $T$ gates, which imply $\epsilon_\texttt{Q}=10^{-5}$. 
The target relative precision (or, equivalently, the absolute precision for $\expval{s\vert U_\Sigma\vert s}$ in Eq. \eqref{eq:jones}) was set to $\varepsilon=10^{-2}$ and we used the minimal basis $\widetilde{B}^{(2)}$ presented in Table~\ref{tab:bases} of the Methods, which has a worst-case negativity $c_\star\approx156$. 
We compare CIPEC (with equal compilation and logical error overheads, \emph{i.e.}, $\omega_1=\omega_2=e$ ensuring $\gamma^2\le e^2$ by Thm.~\ref{thm:main}) with the QEC-only strategy (with equal compilation, statistical, and logical errors, \emph{i.e.}, $\eta_1=\eta_2=1/3$ in Lem.~\ref{lemma:maxG_ECstrategy}). 
The required compilation errors per gate are $\epsilon_{c,i}\approx3.5\times10^{-4}$ and $\epsilon_{c,i}\approx3.7\times10^{-4}$, respectively, which using our {\sc{gridsynth}}-based compilation strategy (see Supplementary Note 4) 
gave compiled circuits of similar sizes, namely $L=3902$ for QEC and $L=3962$ for CIPEC in the worst-case.  
The total number of samples for CIPEC was $S\approx 5.9\times10^{5}$, which corresponds to a true realized overhead of $\gamma^2\approx2.46$ (below the upper bound $e^2$ guaranteed by the Hoeffding inequality); for QEC we used $S\approx1.07\times10^{6}$ samples as prescribed in Lemma 5, which corresponds to an overhead of $\eta_2^{-2}=9$. 
We see that CIPEC is able to deliver an estimate within the desired precision $\varepsilon$ while QEC alone is not (even though we even allowed QEC to use more samples). 
This is in agreement with the $\varepsilon$-independence of the CIPEC overheads and the feasibility conditions illustrated in Fig. \ref{fig:depthmax}. 
We note that the abovementioned choice of parameters gives $L\,\epsilon_\mathtt{Q}\approx3.9\times10^{-2}$ for both CIPEC and QEC, which violates by an order of magnitude the feasibility conditions of Thm \ref{thm:main} and Lemma \ref{lemma:maxG_ECstrategy} (respectively, $L\epsilon_\mathtt{Q}\lesssim3.5\times10^{-3}$) and  $L\epsilon_\mathtt{Q}\lesssim6.7\times10^{-3}$), indicating the correctness of CIPEC even beyond the worst-case bounds guaranteed by Thm \ref{thm:main}.  
Notice also that, even when the QEC-only feasibility condition is satisfied, this strategy still requires a statistical overhead $\eta_2^{-2}$. This overhead can exceed CIPEC's sampling overhead, depending on how much of the total error budget is already consumed by systematic bias from compilation and logical noise.

\vspace{1cm}

\section*{Discussion}

The ability to estimate expectation values using quantum resources that are independent of the target precision can unlock important applications where high precision is required. Apart from the estimation of Jones polynomials exemplified here, other relevant use cases may be found, \emph{e.g.}, in chemistry and materials science.
For instance, for ground-state energy estimation of molecules one typically requires chemical accuracy, which implies that $1/\varepsilon$ grows with the number of spin orbitals. 
For QEC-only strategies, this in turn implies that both the size of the compiled logical circuit and the code-distance (hence also the physical-qubit count) must explicitly grow with the number of orbitals as well. In contrast, CIPEC does the job with both  quantities insensitive to $\varepsilon$, incurring a moderate sample overhead $\gamma^2$. Ultimately, of course, any statistical estimation strategy (including CIPEC and QEC alone) will fail for a task demanding too-high precision, given that the total sample complexity (i.e., number of runs) unavoidably grows with $1/\varepsilon^2$. However, for early fault-tolerant hardware, sample complexity is a far more abundant resource than logical-circuit complexity or code-distance. This means that there is a precision regime achievable only by CIPEC (see Lem. \ref{lemma:maxG_ECstrategy} and Corollary \ref{cor:unsolvable}, in Supplementary Note 2. 
With that in mind, our findings offer a practical route towards fault-tolerant quantum computation with precision-independent overheads.

For simplicity, our analysis assumes stationary gate-dependent logical noise, i.e., that the effective logical error channel associated with each implementable operation depends only on the operation itself. In a detailed fault-tolerant implementation, however, this assumption need not hold. For example, in surface-code architectures based on lattice surgery, the effective logical channel associated with a nominal gate may depend on the circuit schedule, including routing constraints, magic-state availability, and the amount of time the participating logical qubits spend idling before the operation is applied \cite{Litinski_2019}.
Nevertheless, the CIPEC construction itself is not restricted to stationary noise. It remains valid for context-dependent logical error channels, provided that the corresponding effective implementable operations are characterized with sufficient accuracy in diamond norm. This may increase the characterization overhead, since different circuit contexts can, in principle, define distinct effective channels. However, assuming the idling errors are stationary and given that CIPEC tolerates bounded inaccuracies in gate characterization, this restriction can be relaxed by characterizing idling errors and tracking schedule-dependent waiting times.  

While the present work focused on a rigorous analysis of the potential advantage of QEM under abstract (QEC-code- and compiler-agnostic) logical-level assumptions, assessing practical usefulness in a realistic early fault-tolerant device is an important next step. This requires committing to a concrete QEC code family and architecture, such as decoder, logical gate implementations, syndrome-extraction schedule, physical connectivity, correlated logical noise mechanisms, and leakage handling. We leave this analysis for future work.

\vspace{1cm}

\section*{Methods}

\subsection*{Basic setup}

Here we present the basic setup required for the CIPEC algorithm. 
The scenario to keep in mind here is that of early fault tolerant quantum computing, where physical error rates are below threshold but the number of available physical qubits is limited. Namely, assume we are able to realize (and operate on) logical qubits of a particular QEC code but only with a moderate code distance, such that error rates are better than physical ones but not yet to the point of reaching full fault tolerance. 
The idea is to mitigate the leftover errors at the level of the encoded qubits.
First we formally introduce the notion of a noisy logical quantum device as follows. 

\begin{definition}[Noisy logical quantum device]\label{def:noisydevice}
 A noisy logical quantum device $\mathtt{Q}$ relative to a chosen QEC code consists of: 
\begin{enumerate}[leftmargin=*]
    \item  $n_p$ physical qubits encoding $n_\ell\le n_p$ logical qubits;
    \item a set $\widetilde{A}\coloneq\widetilde{V}\cup\widetilde{\Gamma}$ of noisy logical channels whose elements approximate ideal operations from a set $A\coloneq V\cup\Gamma$ of completely positive (CP) maps that contains a universal unitary gate set $V$ and non-unitary operations $\Gamma$.
    We refer to these as \emph{implementable operations} and denote by 
    \begin{equation}\label{eq:noisychannels}
        \epsilon_{\mathtt{Q}} \coloneq \max_{j \in [\vert\mathcal{A}\vert] } 
        \big\|\widetilde{\mathcal{A}}_j-\mathcal{A}_j\big\|_\diamond 
    \end{equation}
    their worst-case diamond-norm error.
\end{enumerate}
\end{definition}
\noindent The parameter $\epsilon_\mathtt{Q}$, which depends on the set of implementable operations, the QEC code distance $d$, and the physical error rate $p$, will be used to express whether or not an instance of Problem~\ref{problem} is solvable by CIPEC on a given device (see Thm.~\ref{thm:main}).  
The non-unitary operations $\Gamma$ may include, \emph{e.g.}, state preparation channels, state projectors, and projective measurements. As we shall see in the following, in the presence of noise these are needed to span the space of two-qubit unitaries. 
Without loss of generality, when considering the noisy version $\widetilde{\Gamma}$ of $\Gamma$ we assume that state preparation and measurement (SPAM) errors are negligible as they can always be incorporated by adding the relevant noisy measurement and noisy state-preparation channels to the set of implementable operations. 
Notice that this definition assumes that implementable logical operations are fully characterized. The cost of characterization is addressed in the Methods.

Implicitly, Def.~\ref{def:noisydevice} assumes stationary gate-dependent errors, meaning that each $\widetilde{\mathcal{A}}_j$ is characterized by the ideal operation $\mathcal{A}_j$ alone, regardless of its position in a circuit or the gates preceding it. This restriction is not strictly necessary and the construction applies to more general noise models as long as they are characterized in diamond norm.  
The implementable logical operations include effective logical error channels $\mathcal{E}_j$, \emph{i.e.} $\widetilde{\mathcal{A}}_j=\mathcal{E}_j\circ \mathcal{A}_j$. 
Although Def.~\ref{def:noisydevice} assumes these errors are perfectly characterized, in Supplementary Note 3 
we discuss the effect of imperfect characterization. 
The $\mathcal{E}_j$ can appear as the result of physical gates and a few cycles of error correction that project the physical state back to the logical space. 
Moreover, we assume $\mathcal{E}_j$ to act locally in at most a small neighbourhood of the qubits acted on by $\mathcal{A}_j$. 
For the unitary gates in $V$, common physical error models manifest as stochastic Pauli noise in the logical level \cite{beale2018quantum}.

The set of universal gates $V$ is determined by the particular QEC code. Here, we focus on $V\coloneq\{I,H,S,S^\dagger,X,Y,Z,T\}^{\otimes2}\cup\{\text{CNOT}\}$, which is a typical gate set for stabilizer codes \cite{terhal2015quantum}. 
The QEC code also determines the relation between the error suppression and the number of physical and logical qubits. 
A popular example is that of distance-$d$ surface codes, for which (for a physical error rate $p$ below the code's threshold $p_{\text{th}}$) the logical error rate is suppressed exponentially in $d$ as $\epsilon_{\mathtt{Q}}\propto\,(p/p_{\text{th}})^{(d+1)/2}$ by employing $2d^2-1$ physical qubits \cite{beverland2022assessingrequirementsscalepractical} per logical one. 
Similar relations hold for other codes (\emph{e.g.}, quantum LDPC codes \cite{Bravyi_2024}). 
Notice that our requirements refer to diamond norm errors, while the gate error $p$ is frequently reported as the average gate infidelity, which can be efficiently estimated via randomized benchmarking. 
However, the two are related as $(1+1/m)\,p\leq\epsilon_{\diamond}\leq\sqrt{(m+1)\,m\,p}$ \cite{Sanders_2015}, where $m$ is the dimension of the space acted on by the error channel, here assumed to be composed of two qubits only. 
Moreover, the lower bound is saturated for stochastic Pauli noise \cite{Sanders_2015}, which is often assumed for obtaining the error threshold of QEC codes \cite{fowler_high_threshold_2012,Bravyi_2024} 
\footnote{Importantly, Pauli noise is known to be a good effective logical error model \cite{beale2018quantum}. 
Moreover, the expression $\epsilon_{\mathtt{Q}}\propto\,(p/p_{\text{th}})^{(d+1)/2}$ gives the error rate per EC cycle to store a qubit. 
In turn, the error per logical gate can be approximated as the number of cycles the gate takes to be implemented ($2d$ cycles for the CNOT and Hadamard gates, for example) times the error per cycle \cite{fowler_surface_2012,suzuki_quantum_2022}.  The proportionality factor depends on the device and code specifications, including the physical error model.}.

Importantly, in a scenario where the number of available physical qubits is limited, logical errors cannot be arbitrarily suppressed. 
In solving Problem \ref{problem}, this implies the target precision $\varepsilon$ may not be achievable through QEC alone even if statistical fluctuations were absent, in which case combining error-correction with error mitigation methods can prove useful. 
We resort to probabilistic error cancellation (PEC), a technique proposed to mitigate errors in the context of NISQ devices \cite{Temme_2017,Takagi_2021}. 
Previous approaches to PEC employ the so called error inversion method \cite{suzuki_quantum_2022, Jin_2025}, where the inverse noise channels $\mathcal{E}_j^{-1}$ are decomposed as linear combinations of noisy implementable channels (assuming the noise is invertible) and applied after each noisy gate in the circuit, which may cause a significant increase in circuit depth. 
Here, instead, we consider the compensation method for PEC \cite{endo2018errormitigation}, in which case the ideal gates $U_i$ in the circuit are decomposed in a similar way. 
For that to be possible, one must ensure that the implementable operations span the space of all $2$-qubit unitary channels (as opposed to the inversion method, which in general requires the larger space of all completely positive trace-preserving maps). 
Notice that, for arbitrary noise models $\mathcal{E}_j$, composing the noisy unitary operations $\widetilde{V}_j$ alone may not suffice, and we need to include non-unitary operations. 
This is formalized in the following definition.

\begin{definition}[CIPEC-capable device]\label{def:feasible} 
Let $\mathtt{Q}$ be a noisy logical quantum device as in Def. \ref{def:noisydevice}, ${V}^{(D)}\subset \left\{\prod_{i=1}^D{\mathcal{V}}_{\alpha_i}\,\vert \,\alpha_1,\ldots,\alpha_D\in[\vert V\vert]\right\}$ be a set of unitary channels obtained by composing $D >0$ operations from the universal set $V$, and $\widetilde{V}^{(D)}$ be the implementable version of $V^{(D)}$. 
We say that $\mathtt{Q}$ is capable of running CIPEC if there exists $\widetilde{V}^{(D)}$ such that the span of $\widetilde{B}\coloneq\widetilde{V}^{(D)}\cup\,\widetilde{\Gamma}$ contains the space of $2$-qubit completely positive and trace preserving (CPTP) maps on the logical space. 
We refer to $\widetilde{B}$ as the basis of implementable operations (or simply \emph{noisy basis}) and $D$ as the maximum sequence length (or depth) of $\widetilde{B}$. The corresponding set ${B}\coloneq{V}^{(D)}\cup\,{\Gamma}$ of ideal operations will be referred to as the \emph{noiseless basis}.
\end{definition}

\noindent Def.~\ref{def:feasible} should be viewed as a strong sufficient condition for the applicability of CIPEC. 
Whether or not a given noisy logical device is CIPEC-capable depends heavily on the noisy logical operations available on the device. 
In practice, however, for a properly characterized logical device this property is easily verifiable: one may enumerate implementable operations up to depth \(D\), form their Choi representations, and check if the rank of the resulting Gram matrix equals the dimension of the space of CPTP maps, $\text{dim(CPTP)}=241$ (see Supplementary Note 1 
for details). 

Notice that the basis $\widetilde{B}$ in Def.~\ref{def:feasible} can be overcomplete. 
In Supplementary Note 1 
we discuss the effects of the basis size $\vert\widetilde{{B}}\vert$ and maximal sequence length $D$ on practical implementation. 
The following definition introduces the negativity, which accounts for the sample overhead of error mitigation, and its upper bound $c_\star $, that will play an important role in the analysis of CIPEC.

\begin{definition}[Worst-case negativity] \label{def:negativity}
Given a basis $\widetilde{B}$ of implementable operations as in Def. \ref{def:feasible}, the negativity of a channel $\mathcal{C}\in\text{span}(\widetilde{B})$ in this basis is 
$\|\mathcal{C}\|_{\widetilde{B}} := \min_{b}\left\{\sum_{j}|b_{j}|
\text{  s.t. }\mathcal{C} = \sum_{j\in[\vert \widetilde{B}\vert]} b_j\,\widetilde{\mathcal{B}}_j\right\}$. 
We denote by $c_\star $ the worst-case negativity over all the \emph{unitary} channels $\mathcal{C}$,
\begin{align}
    c_\star \coloneq \max_{\norm{\mathcal{C}}_\diamond=1} \norm{\mathcal{C}}_{\widetilde{B}}\,.
\end{align}
\end{definition}
\noindent In Supplementary Note 1 
we prove that, for a generic CPTP channel $\mathcal{C}$, $c_\star$ connects the negativity and the diamond norm as $\|\mathcal{C}\|_{\widetilde{B}}\leq c_\star \|\mathcal{C}\|_\diamond$ and that the worst-case negativity over the full space of CPTP channels $\mathcal{C}$ is attained at the subset of unitary channels and, therefore, is also given by $c_\star$. We use this property in the proof of Thm.~\ref{thm:main} (see Supplementary Note 2) for bounding the negativity of the residual channels  $\mathcal{U}_i - \widetilde{\mathcal{V}}_{c,i}^{(L_{i})}$ appearing in Eq.~\eqref{eq:lp}.

\vspace{1cm}

\subsection*{Three different bases}

A minimal channel basis for two-qubit CP maps was proposed in \cite{endo2018errormitigation} using Clifford gates conjugated with trace non-increasing  channels $\pi_z(\rho)\coloneq\left(\frac{1+Z}{2}\right)\rho\left(\frac{1+Z}{2}\right)$ that project any state $\rho$ into the $\ketbra{0}{0}$ state, but these make it inconvenient for PEC due to the need for postselection. 
Here, we introduce two different bases whose span contain the space of CPTP maps and investigate their influence on the statistical overhead of PEC as measured by the constant $c_\star $ of Def. \ref{def:negativity}. 
For each noisy basis $\widetilde{B}$, we empirically estimate $c_\star $ by taking the worst-case negativity over an ensemble of $10^4$ Haar random unitaries. 
In particular, we explicitly construct bases from state preparation channels $\mathcal{P}_{\ket{\psi}}(\rho)\coloneq\ketbra{\psi}{\psi}$, which prepare the state $\ket{\psi}$ from any given $\rho$, and sequences of Clifford operations only, with no need for $T$ gates, so that all its unitary elements are classically simulable. 
We first construct a minimal basis of maximal sequence length $D=4$ using sequences of implementable gates and state preparations. 
From all possible sequences of implementable operations, we use a greedy search to choose $241$ linearly independent elements to form a minimal basis spanning the space of 2-qubit channels. 
We also investigate the advantage of using an overcomplete basis given by the full $2$-qubit Clifford group and state preparation channels, which reduces the negativity of the decomposition but uses longer sequences with $D=17$. 
The main features are summarized in Tab. \ref{tab:bases}. 
Since the runtime and performance guarantees of CIPEC depend on the chosen basis, an open question is whether one can go beyond the above heuristics and construct an optimal basis, \emph{i.e.}, one with the smallest value of $c_\star $ while having the minimum number of elements.

\begin{table}[t]
    \renewcommand{\arraystretch}{1.3}
    \begin{tabularx}{\columnwidth}{@{\extracolsep{\fill}\hspace{0pt}}|c|cc|ccc|}
      \hline
      \textbf{$\widetilde{B}$} &  $\widetilde{V}^{(L)}$ &  $\widetilde{\Gamma}$ &  $\vert \widetilde{B}\vert$ & $D$ & $c_\star $ \\
      \hline
      $\widetilde{B}^{(1)}$ & Clifford group & State prep. & $11535$ & $17$ & $4.47$ \\
      \hline
      $\widetilde{B}^{(2)}$ & Cliffords (min) & State prep. & $241$ & $4$ & $156.2$ \\
      \hline
      $\widetilde{B}^{(3)}$ & Cliffords (min) & State proj. & 256 & $10$ & $88.0$ \\
      \hline
    \end{tabularx}%
  \caption{\textbf{Three different bases for CIPEC.} 
  The noisy unitary basis elements $\widetilde{V}^{(L)}$ use only Clifford gates $\{I,H,S,S^\dagger,X,Y,Z,\text{CNOT}\}$. 
  The span of basis $\widetilde{B}^{(1)}$ contains the space of completely positive and trace-preserving (CPTP) maps for $2$-qubits and the basis is overcomplete. It consists of the full $2$-qubit Clifford group plus state preparation channels $\widetilde{\Gamma}=\{I,\mathcal{P}_{\ket{+}},\mathcal{P}_{\ket{+y}},\mathcal{P}_{\ket{0}}\}^{\otimes2}$ that prepare the corresponding state $\ket{\psi}$ from any $\rho$. 
  The basis $\widetilde{B}^{(2)}$ uses the same state preparation channels $\widetilde{\Gamma}$, but with a minimal number of linearly independent compositions of Clifford gates necessary so its span contains the space of CPTP maps, with $\text{dim}(\text{CPTP})=241$ \cite{github_cipec}. 
  The basis $\widetilde{B}^{(3)}$ was proposed in \cite{endo2018errormitigation} (see also \cite{Takagi_2021}) and is also minimal, using Clifford gates conjugated with trace non-increasing  channels $\pi_z$ that project $\rho$ into the $\ket{0}$ state. 
  Its span contains the space of CP maps, with $\text{dim}(\text{CP})=256$. 
  The reported values of $c_\star $ assume a local depolarizing noise model after each operation with strengths $10^{-6}$ for products of single-qubit gates and $10^{-5}$ for the others, leading to $\epsilon_\mathtt{Q}=10^{-5}$.
  }
    \label{tab:bases}
\end{table}

\vspace{1cm}

\subsection*{Logical gate characterization}
\label{sec:Logical_gate_char}

To characterize the implementable operations in  $\widetilde{\mathcal{A}}$, we can use standard gate set tomography (GST) methods \cite{merkel2013,blumekohout2013robustselfconsistentclosedformtomography}. 
$\widetilde{A}$ includes noisy gates $\widetilde{V}$ from a universal set and the non-unitary channels $\widetilde{\Gamma}$. 
The former are well described by ideal unitary gates followed by stochastic Pauli noise, while the latter can be more general. 
For concreteness, here we focus on a device that can implement CNOT and tensor-product gates $\{I,H,S,S^\dagger,T\}^{\otimes2}$ on any pair of logical qubits connected by the fault-tolerant architecture in question (Pauli gates ($X,Y,Z$) can be implemented without noise by changing the Pauli frame \cite{Knill_2005}). The gate-set size is then $\vert\widetilde{A}\vert=26\tau$, where $\tau$ is a logical-architecture dependent factor that, in the foreseen fault-tolerant architectures, is expected to scale as $\tau=\mathcal{O}(n)$. For example, in superconducting-qubit platforms \cite{yoder2025tourgrossmodularquantum}, the scaling follows from local circuit connectivity. In turn, for platforms with all-to-all connectivity due to movable qubits, such as Rydberg atoms \cite{QuEra_FT_architecture}, 2-qubit gates will be executed at few entangling zones while single-qubit ones will be executed at each logical qubit's position, giving again $\tau=\mathcal{O}(n)$. 
Standard GST requires $\mathcal{O}(1/\epsilon_\text{char}^2)$ uses to estimate a two-qubit channel up to diamond norm error $\epsilon_\text{char}$ and, as shown in Thm. 7 in the Supplementary Note 3, CIPEC requires $\epsilon_\text{char}\le\varepsilon/(2L\|O\|)$ for each gate.  
Putting everything together, this implies a total of
$\mathcal{O}\left(\tau L^2\|O\|^2/\varepsilon^2\right)$
runs, which is prohibitive for practical applications. 
Fortunately, this scaling can be improved using long-sequence GST \cite{Nielsen_2021}. 
Empirically, it has been verified that, using a set of $\mathcal{O}(\log(L_{\text{GST}}))$ sequences of gates with maximum depth $\mathcal{O}(L_{\text{GST}})$ and collecting $S$ measurement samples for each, the diamond norm error in the estimated gates scales as $\epsilon_\text{char}=\mathcal{O}(1/(L_{\text{GST}}\sqrt{S}))$. 
Therefore, the time required by long-sequence GST to characterize a given 2-qubit gate set is $T_\text{char}=\mathcal{O}(\log(L_{\text{GST}})\,S\,L_{\text{GST}})$ per logical qubit pair. 
The larger the sequence size, the more precise the estimation, as long as the hardware allows it (i.e., $L_{\text{GST}}=\mathcal{O}(1/\epsilon_{\mathtt{Q}})$). 
Therefore, considering sequences of maximum depth $L_{\text{GST}}=L$ and $S=\mathcal{O}(\tau/\epsilon^2)$ samples yields a time $T_\text{char}=\mathcal{O}(\tau L\log(L)/\varepsilon^2)$ to characterize all the implementable operations up to the precision required to guarantee stability of CIPEC. 
Moreover, in case the platform allows for parallel execution of 2-qubit gates, this scaling can be further reduced to $T_\text{char}=\mathcal{O}({(\tau/n)} L\log(L)/\varepsilon^2)$. For instance, for 1D connectivity, the universal set can be split into $n-1$ gate sets, which can be characterized in only two separate experiments. In turn, for the 2D square lattice, gates on $2(n-\sqrt{n})$ pairs of qubits have to be characterized, which can be done in $4$ experiments by first scanning row connections and then column ones. 
To assess the feasibility of $T_\text{char}$, we compare it against CIPEC's observable estimation time, namely $T_\text{est}=\mathcal{O}\big((L/n)\gamma^2\|O\|^2/\varepsilon^2\big)$, since the largest circuit sampled has $L$ gates and up to $\mathcal{O}(n)$ of them can be run in parallel. 
We thus see that $T_\text{char}/T_\text{est}=\mathcal{O}\big(\tau \log(L)/(\gamma^2\|O\|^2)\big)$, which scales linearly with the number of qubits ($\tau=\mathcal{O}(n)$) and mildly with $L$. 
Therefore, characterization can be done with a modest runtime overhead compared to CIPEC's own runtime. 

In the analysis above, we have considered that each gate in the gate set has a fixed error channel to be characterized. However, in a realistic fault-tolerant scenario, the effective error can depend on the context in which the gate is implemented. Nevertheless, a significant source of logical errors comes from schedule-dependent idling times \cite{Litinski_2019}. Under the common assumption of Markovian error,  this is already included in the analysis when we consider the characterization of the identity channel over one error correction cycle.   

\vspace{1cm}

\subsection*{Finding the quasi-probability decompositions} \label{app:lp}

The optimization problem in Eq. \eq{lp} has $v=|\widetilde{B}|$ variables and $c=256$ constraints (for $2$-qubit gates) corresponding to the matrix entries of $\mathcal{U}_i$ and can be solved using standard convex optimization solvers \cite{ben2001} in time $\text{poly}\big(v,c\big)=\text{poly}\big(|\widetilde{B}|\big)$. 
We used {\sc cvxpy} \cite{diamond2016cvxpy,agrawal2018rewriting} as the modeling interface to the {\sc mosek} \cite{mosek} solver. 
We chose to keep $b_{i,\vert \widetilde{B}\vert+1}\coloneq1$ out of the optimization variables since this allows simpler proofs when bounding the total negativity in Lemma \ref{lemma:negativity_ciPEC}, and a clear choice of compilation error in step {\it 1. (a)}. 
In our numerics, a solution in the largest basis $B^{(1)}$ of $\sim11$k elements took about $5$ minutes on a standard laptop, while for bases $B^{(2)}$ and $B^{(3)}$ the time was about $2\,$ms per gate. In fact, the completeness of the last two bases ensures the existence of a unique decomposition for each channel that can be obtained via linear system solving, which is typically cheaper than solving the LP. We therefore observe a trade-off between classical overhead and quantum resource requirements, controlled by the worst-negativity $c_\star$, motivating the search for bases with a small number of elements and a small $c_\star$ (see Thm.~\ref{thm:main} and Table~\ref{tab:bases}). 

\vspace{1cm}

\subsection*{Two-qubit gate compilation} \label{app:compiler}

For the compilation step, the method supports any two-qubit gate synthesis algorithm.  
Here we illustrate it for a strategy as follows: we first decompose $U_i$ into CNOT, $\sqrt{X}$, and $R_z$ gates, and then use {\sc{gridsynth}} \cite{ross2016} to synthesize each $R_z$ into Clifford + $T$ gates, from which we extract the sequence lengths $L_i$ appearing in Theorem \ref{thm:main}. This is described in detail in Supplementary Note 4 
where, as a side result which may be of independent interest, we also numerically estimate the scaling with the diamond-norm error $\epsilon_c$ of average sequence length $L$ of this compilation strategy for Haar random $2$-qubit unitaries. 
We observe the Solovay-Kitaev-like scaling 
\begin{align}\label{eq:SK}
L_\text{SK}(\epsilon_\text{c})=c_1\log^{c_2}\big(1/\epsilon_c\big),
\end{align}
where $c_1\approx210$ and $c_2\approx0.75$. 
Improved compilation strategies such as probabilistic synthesis \cite{akibue2024,Campbell_2017} can lead to better constants: \emph{e.g.}, using a probabilistic sum of distinct compilations of precision $\mathcal{O}(\epsilon_c)$ reduces the diamond norm error to $\mathcal{O}(\epsilon_c^2)$ \cite{Campbell_2017}, which effectively improves the value of $c_1$ by a factor of $2^{c_2}\approx1.7$. The use of multiple compilations does not necessarily lead to a proportional increase in classical preprocessing cost. In particular, for decompositions based on $R_z$
 gates, the different terms can be generated from a single compiled sequence by appending appropriate correction gates \cite{Khitrin_2026}. Moreover, in analogy to how the use of multiple transpilations of a gate can reduce the negativity of PEC on NISQ devices \cite{hsieh2024}, the use of distinct compilations in the affine decomposition of each gate also has the potential to improve $\gamma$.   

\vspace{1cm}

\subsection*{Quantum circuits for Jones polynomial estimation}
\label{app:jones}

Here we briefly review some basic properties of knots and the control-free Hadamard test circuits introduced in \cite{laakkonen_less_2025} for estimating the Jones polynomial \cite{Kauffman2001} of a knot on a quantum computer. 

Recall that any knot can be represented as the closure of some braid $\Sigma$ \cite{Kauffman2001,Melnikov2024} on a number $n_\text{s}$ of strands. 
Any such braid can be expressed as a product of braid generators $\sigma_i^{\pm1}$ acting on the $(i,i+1)$ pair of strands, where $i\in[n_\text{s}-1]$ and the number $\vert \Sigma\vert$ of terms in the product gives the total number of braid crossings. 
These braids can be closed in multiple ways, the most used ones being the so-called Markov and plat closures, $\text{M}(\Sigma)$ and $\text{P}(\Sigma)$ respectively \cite{Birman1976,Melnikov2024}. 
The former exists for any number of strands and is constructed by connecting the endpoints of each strand in $\Sigma$ through simple non-crossing arcs; the latter, on the other hand, requires a braid $\Sigma$ with an even number of strands and is obtained by connecting pairs of endpoints sequentially (see Fig. 1 in \cite{laakkonen_less_2025} for an illustration). 
For instance, the trefoil knot illustrated in Fig. \fig{jones} can be represented either as $\text{M}(\Sigma)$ for the $2$-strand braid $\Sigma=\sigma_1^3$ or as $\text{P}(\Sigma)$ for the $4$-strand braid $\Sigma=\sigma_1\,\sigma_2^3\,\sigma_1^{-1}$. 

The Jones polynomial of a knot (or link) $K$, denoted $J_K(q)$, is a polynomial over a complex variable $q$ with integer coefficients determined solely by the topology of $K$. 
We focus on the estimation of $J_K(q)$ at the special point $q=e^{2\pi i/5}$ to relative precision $\varepsilon$. 
For $K=\text{M}(\Sigma)$ and $K=\text{P}(\Sigma)$ this problem is known to be DQC1-complete \cite{shor_estimating_2008} and BQP-complete \cite{Aharonov2009}, respectively, for $\varepsilon=1/\mathcal{O}(\text{poly}(n_\text{s}))$ and $\vert \Sigma\vert=\mathcal{O}(\text{poly}(n_\text{s}))$. 
For simplicity, here we focus on the BQP-complete case. 
As shown in \cite{laakkonen_less_2025}, the Jones polynomial can be expressed as 
\begin{align}\label{eq:jones}
J_{\text{P}(\Sigma)}\big(e^{i \frac{2\pi}{5}}\big) = 
\big(-e^{-i\frac{3\pi}{5}}\big)^{3\omega_\Sigma}\phi^{\frac{n_\text{s}}{2}-1}\expval{s\vert U_\Sigma\vert s}\,,
\end{align}
where $\phi\coloneq\frac{1}{2}(1+\sqrt{5})$ is the golden ratio, $\omega_\Sigma\coloneq\sum_{\sigma\in \Sigma}\text{sign}(\sigma)$ is called the writhe of the braid $\Sigma$, $\ket{s}\coloneq \ket{0101\cdots010}$ is $(n_\text{s}+1)$-qubit a computational basis state, and $U_\Sigma$ is the Fibonacci unitary representation of the braid $\Sigma$ \cite{kauffman_fibonacci_2008} (see \cite{laakkonen_less_2025} for a specific realization in terms of single- and two-qubit gates). 
We use the quantum algorithm proposed in \cite{laakkonen_less_2025}, which is based on a Monte Carlo estimation of the matrix element $\expval{s\vert U_\Sigma\vert s}$ in Eq. \eqref{eq:jones} using a control-free Hadamard test. 
The associated circuit $U$ (see \cite{laakkonen_less_2025}) contains $3\vert \Sigma\vert$ parametrized two-qubit gates, $\mathcal{O}(n)$ controlled-NOT gates, and $\mathcal{O}(\vert \Sigma\vert)$ single-qubit rotations.  
This falls within the scope of Problem \ref{problem} with $U\rho \,U^\dagger=\ketbra{s}{s}$, $G=\mathcal{O}(\vert\Sigma\vert+ n)$ two-qubit gates, and $O=\frac{1}{2}(U_\Sigma^{\phantom{\dagger}}+U_\Sigma^{{\dagger}})$ or $O=\frac{1}{2i}(U_\Sigma^{\phantom{\dagger}}-U_\Sigma^{{\dagger}})$ for the real and imaginary parts of $\expval{s\vert U_\Sigma\vert s}$, respectively. 
After estimating $\expval{s\vert U_\Sigma\vert s}$, the Jones polynomial estimate is obtained by post-processing the result by the factor in Eq. \eqref{eq:jones}. 
Notice that, since each of these gates needs to be compiled into a sequence of $L_i=\mathcal{O}\left(\log^{c_2}(G)\right)$ universal gates from the set $V$, the total circuit size is $L=\mathcal{O}\!\left(G\log^{c_2} G\right)=\mathcal{O}\big((\vert\Sigma\vert+n)\log^{c_2}(\vert\Sigma\vert+n)\big)$. 

In the example of the trefoil knot reported in Fig. \ref{fig:jones}, we group single- and two-qubit gates acting on the same target qubits into a new two-qubit gate $U_i$, so that the resulting circuit $U$ is composed of only $G=9$ such $U_i$ acting on distinct qubit pairs.

\vspace{.2cm}

\section*{Author Contributions}
G.C., T.O.M., A.T., T.L.S., D.S.F., and L.A. proved the main analytical results. 
G.C., T.O.M., A.T., and A.A. performed the numerical simulations. 
All authors contributed to the conception and write-up of the
paper. All authors have read and approved the manuscript.

\vspace{.2cm}

\section*{Data availability}
The data used in this study was generated using the code available on \url{https://github.com/giancamilo/CIPEC}. 

\vspace{.2cm}

\section*{Code availability}
The code used to generate the data in this study is available on \url{https://github.com/giancamilo/CIPEC}. 

\vspace{.2cm}

\section*{Competing Interests}
The authors declare no competing financial or non-financial interests.
\vspace{.2cm}

\section*{Acknowledgements}
We thank Ryuji Takagi, Suguru Endo, Fernando G. S. L. Brandão, Samson Wang, Ingo Roth, and Ariel Bendersky for discussions. 

\vspace{.2cm}

\section*{Funding}
D.S.F.'s work was funded by the Novo Nordisk Foundation (Grant No. NNF20OC0059939 Quantum for Life).

\vspace{.2cm}

\section*{References}
\bibliography{references}

\newpage

\section*{Supplementary Information}

\appendix

\renewcommand{\thesection}{\arabic{section}}
\setcounter{section}{0}
\setcounter{figure}{0}
\renewcommand{\thefigure}{S\arabic{figure}}

\vspace{1cm}

\section{Building a basis of 2-qubit channels from implementable operations}
\label{app:basis}

Here we show how to build a basis $\widetilde{B}$ of channels in Def. \ref{def:feasible} using the set of implementable operations of a noisy logical quantum device. This procedure was used to build the bases reported in Table~\ref{tab:bases}. Moreover, we also prove that the worst-case negativity parameter $c_\star$ defined in Def.~\ref{def:negativity} connects the negativity $\norm{\mathcal{C}}_{\widetilde{B}}$ and the diamond-norm $\norm{\mathcal{C}}_\diamond$ of a generic channel $\mathcal{C}$.

First, consider a two-qubit system $AB$, where $A$ labels the first qubit and $B$ the second.  
The linear span of the set of all CP maps in $AB$ in its Choi representation is the real vector space $\text{Herm}(AB)$ of Hermitian operators acting on the composite system, whose dimension is $4^{4}=256$. 
The linear span of the set of all CPTP maps, on the other hand, is given by  the subspace of all $X\in\text{Herm}(AB)$ with $\trace_{B}X=\openone_{A}$, which has dimension $4^{4}-4^{2}+1=241$ (see, \emph{e.g.}, \cite{watrous2018}). 
Therefore, given a set of two-qubit (CPTP or CP) noisy channels, one can test whether it is a spanning set for the span of CPTP or CP by computing the rank of the Gram matrix built from its elements under the Hilbert-Schmidt inner product, defined by $\trace(XY)$ for all $X,Y\in\text{Herm}(AB)$.

Given a CIPEC-capable noisy quantum device $\mathtt{Q}$ as in Def. \ref{def:feasible}, it is always possible to build a set $\widetilde{B}$ of noisy channels of the form $\widetilde{V}^{(D)}\cup\,\widetilde{\Gamma}$ that passes the Gram matrix test, where we recall that $\widetilde{V}^{(D)}$ denotes a set of sequences of noisy universal gates from $\mathtt{Q}$ of length up to $D$. 
However, these are far from unique, and some sets of channels might be preferable given specific conditions. 
To assess how suited one set of channels will be for a particular application, we need to identify relevant figures of merit. 
From the description of the CIPEC algorithm (Alg. \ref{alg:CIPEC}), it is natural to worry about:
$i)$ the amount of classical pre-processing required to build and store the classical description of all distinct circuits that need to be estimated on the device;
$ii)$ the length (or depth) of each one of those circuits and;
$iii)$ the total number of samples required for the estimation task.
These concerns lead us to identify three figures of merit for a particular set $\widetilde{B}$, respectively: 
the first is the number of elements $|\widetilde{B}|$, which affects the number of distinct circuits that need to be implemented on a particular task; 
second, the maximal length $D$ of the elements of $\widetilde{B}$ affects the total length and depth of each distinct circuit; 
and lastly, the maximal negativity $c_{*}$ of a unitary channel decomposed with respect to $\widetilde{B}$ indirectly affects the total number of samples required for the estimation. 
Although it is not hard to see why the first two figures of merit address their respective concerns, the last one requires further explanation.

Let $\widetilde{B}$ be a spanning set for the linear span of CPTP (or CP) channels, and let $c_{*}$ be as given in Def. \ref{def:negativity}.
The following lemma shows that $c_{*}$ is the maximal amplification factor required for $\|\mathcal{C}\|_{\diamond}$ of an arbitrary channel $\mathcal{C}$ to bound negativity $\|\mathcal{C}\|_{\widetilde{B}}$, which is related to the sample complexity overhead of error mitigation.

\begin{lemma}[Negativity and the diamond norm] \label{lem:1-norm_diamond} Given a basis $\widetilde{B}$ for the vector space of two-qubit maps, its associated negativity satisfies
\begin{equation}
    \|\mathcal{C}\|_{\widetilde{B}}\leq c_\star  \|\mathcal{C}\|_{\diamond}, \qquad c_\star  \coloneq \max_{\norm{\mathcal{C}}_\diamond=1} \norm{\mathcal{C}}_{\widetilde{B}}\,.
\end{equation}
\end{lemma}
\begin{proof}
Given two norms $\norm{\cdot}_{\widetilde{B}}$ and $\norm{\cdot}_\diamond$ on a finite-dimensional vector space, they must be connected within constant factors of one another. Specifically, there exist two real numbers $0 < c_{\text{low}} \leq c_{\text{up}}\coloneq c_\star $ such that, for all $\mathcal{C}$ in this vector space, we have
\[
    c_{\text{low}}\norm{\mathcal{C}}_\diamond \leq \norm{\mathcal{C}}_{\widetilde{B}} \leq c_{\text{up}}\norm{\mathcal{C}}_\diamond.
\]
This inequality is trivially true for $\mathcal{C}=0$.  For $\mathcal{C} \neq 0$, we can restrict  to the case where $\norm{\mathcal{C}}_\diamond=1$, obtaining
\[
    c_{\text{low}} \leq \norm{\mathcal{C}}_{\widetilde{B}} \leq c_{\text{up}}.
\]
By the extreme value theorem, a continuous function (in this case $\norm{\cdot}_{\widetilde{B}}$) on a compact set (the closed and bounded set of unitary channels defined by $\norm{\mathcal{C}}_\diamond=1$) must achieve a maximum and minimum value on this set, therefore
\begin{gather}
\begin{aligned}
    c_{\text{low}} \coloneq \min_{\norm{\mathcal{C}}_\diamond=1} \norm{\mathcal{C}}_{\widetilde{B}} \\
    c_{\text{up}} \coloneq \max_{\norm{\mathcal{C}}_\diamond=1} \norm{\mathcal{C}}_{\widetilde{B}}.
\end{aligned}
\end{gather}
\end{proof}

\section{Main proofs}
\label{app:proofs}

Here we prove our main technical result, namely Theorem \ref{thm:main}. We also prove two similar statements for the \emph{QEC only} estimator that leverages partial quantum error-correction without quantum error mitigation (Lemma \ref{lemma:maxG_ECstrategy}) and the \emph{standard PEC} estimator (Lemma \ref{lemma:standardPEC}).  
For that, we present the following two auxiliary lemmas. 
The first gives the diamond norm error of two-qubit gate synthesis after replacing ideal gates from the universal set $V$ by noisy gates from the set of implementable operations $\widetilde{V}$. 
The second leverages this to prove an upper bound on the negativity of CIPEC, from which the main theorem then follows straightforwardly. 

\begin{lemma}[Two-qubit Unitary Synthesis using noisy channels]\label{lemma:noisySK}
Let $\mathcal{U}$ be a two-qubit unitary channel, $\mathcal{V}_{c}^{(L)}$ be a synthesis of $\mathcal{U}$ with $L$ ideal channels from the universal set $V$ with compilation error $\big\|\mathcal{U}-\mathcal{V}_{c}^{(L)}\big\|_\diamond\le\epsilon_c$, 
and $\widetilde{\mathcal{V}}_{c}^{(L)}$ be the version of $\mathcal{V}_{c}^{(L)}$ with noisy operations from the implementable set $\widetilde{V}$. 
Then $\big\|\mathcal{U}-\widetilde{\mathcal{V}}_{c}^{(L)}\big\|_\diamond \le \epsilon_c + L\,\epsilon_\mathtt{Q}$.
\end{lemma}

\begin{proof}
\begin{align*}
\big\|\mathcal{U}-\widetilde{\mathcal{V}}_{c}^{(L)}\big\|_\diamond &= \big\|\mathcal{U}-\mathcal{V}_{c}^{(L)}+\big(\mathcal{V}_{c}^{(L)}-\widetilde{\mathcal{V}}_{c}^{(L)}\big)\big\|_\diamond
\\
&\le \big\|\mathcal{U}-\mathcal{V}_{c}^{(L)}\big\|_\diamond + \big\|\mathcal{V}_{c}^{(L)}-\widetilde{\mathcal{V}}_{c}^{(L)}\big\|_\diamond \\
&\le \epsilon_c + L\,\epsilon_\mathtt{Q}\,,
\end{align*}
where $\big\|\mathcal{U}-\mathcal{V}_{c}^{(L)}\big\|_\diamond\le\epsilon_c$ by assumption, and $\big\|\mathcal{V}_{c}^{(L)}-\widetilde{\mathcal{V}}_{c}^{(L)}\big\|_\diamond\le L\,\epsilon_\mathtt{Q}$ follows by repeated use of the inequality $\|\mathcal{C}_1\mathcal{C}_0-\mathcal{D}_1\mathcal{D}_0\|_\diamond \le \|\mathcal{C}_0-\mathcal{D}_0\|_\diamond+\|\mathcal{C}_1-\mathcal{D}_1\|_\diamond$. 
\end{proof}

\noindent This result can then be used to prove the following upper bound on the negativity of CIPEC, which, in turn, is the main ingredient for proving Theorem \ref{thm:main}. 

\begin{lemma}[CIPEC negativity upper bound]\label{lemma:negativity_ciPEC}
Let $\epsilon_\mathtt{Q}$ as in Eq. \eq{noisychannels}, $\mathcal{U}_i$ a two-qubit unitary channel, and $\mathcal{U}_i=\widetilde{\mathcal{V}}_{c,i}^{(L_i)}+ \sum_{j\in[|\widetilde{B}|]} b_{i,j} \widetilde{\mathcal{B}}_j$ be its compilation-informed quasi-probability decomposition Eq. \eqref{eq:lp}. 
Then the total negativity of $\mathcal{U}=\prod_{i\in[G]}\mathcal{U}_i$ is $\gamma \le e^{\lambda}$ with $\lambda:=\sum_{i\in[G]}c_{*}\big[\epsilon_{c,i} + L_i\,\epsilon_\mathtt{Q}\big]$. 
\end{lemma}

\begin{proof}
For each $i\in[G]$ the minimum negativity satisfies
$\gamma_i = 1+\norm{\mathcal{U}_i-\widetilde{\mathcal{V}}_{c,i}^{(L_i)}}_{\widetilde{B}}\le1+c_{*}\norm{\mathcal{U}_i-\widetilde{\mathcal{V}}_{c,i}^{(L_i)}}_{\diamond}$, according to Lemma \ref{lem:1-norm_diamond}. 
Using Lemma \ref{lemma:noisySK} we get $\gamma_i\le 1+\lambda_i$ with $\lambda_i=c_\star [\epsilon_{c,i} + L_i\,\epsilon_\mathtt{Q}]$, and hence $\gamma=\prod_{i\in[G]}\gamma_i\le\prod_{i\in[G]}(1+\lambda_i)\le e^{\lambda}$ with $\lambda\coloneq\sum_{i\in[G]}\lambda_i$. 
\end{proof}

\noindent With the results above, we can finally prove our main statement, Theorem \ref{thm:main}, which we restate here for convenience 

\begin{reptheorem}{thm:main}
Let $\{U_i\}_{i\in[G]}, \varepsilon,\delta,\rho, O$ as in Problem \ref{problem},  $\mathtt{Q}$ be a CIPEC-capable device with worst-case error $\epsilon_{\mathtt{Q}}$ and worst-case negativity $c_\star $, and $\omega_1,\omega_2>1$ be constants to be used, respectively, for parameterizing the compilation and logical error contributions to the total sample overhead. 
Denote by $L_i$ the length of a logical-gate sequence that compiles $U_i$ up to diamond-norm error $\epsilon_{c,i}\le\log(\omega_1)/(2c_\star G)$ into the universal set $V$ realized by $\mathtt{Q}$, and let $L \coloneq\sum_{i\in[G]}L_i$ be the total number of logical gates. 
Then, if $L\,\epsilon_\mathtt{Q}\le \log(\omega_2)/(2c_\star \epsilon_\mathtt{Q})$, Alg. \ref{alg:CIPEC} solves Problem \ref{problem} on $\mathtt{Q}$ for any target precision $\varepsilon$ with constant sample overhead $\gamma^2\le\omega_1\omega_2$, i.e., using at most $\frac{1}{2}\gamma^2\,\norm{O}^2\varepsilon^{-2}\log(2/\delta)$ samples.
\end{reptheorem}

\begin{proof}
Since, by construction, the quasiprobability decomposition in Eq. \eqref{eq:lp} is an unbiased estimator for each $\mathcal{U}_{i}$, the total estimation error of the expectation value in Alg. \ref{alg:CIPEC} comes only from statistics and can be arbitrarily decreased by collecting more samples. 
By Hoeffding inequality, a precision $\varepsilon$ can be achieved (with probability $1-\delta$) by collecting 
\begin{align}\label{eq:samples-cipec}
 S_{\text{CIPEC}}\coloneq\gamma^2\|O\|^2\log(2/\delta)/(2\varepsilon^2)   
\end{align}
samples, where $\gamma$ is the total negativity of the quasiprobability decomposition of $\mathcal{U}$, $\|O\|$ is the spectral norm of $O$, which gives the range of the random variable under our assumption of  measuring in the eigenbasis of $O$. 
In Lemma \ref{lemma:negativity_ciPEC}, we showed that $\gamma\le e^{\lambda}$ with $\lambda:=\sum_{i\in[G]}c_{*}\big[\epsilon_{c,i} + L_i\,\epsilon_\mathtt{Q}\big]$.
This can be made constant, more precisely $\gamma\le \sqrt{\omega_1\omega_2}=\mathcal{O}(1)$, by ensuring $\lambda\le\log(\sqrt{\omega_1\omega_2})$ as follows: 
$i)$ compile each $\mathcal{U}_i$ to precision $\epsilon_{c,i}\le\log(\sqrt{\omega_1})/(c_\star G)$ such that $\sum_{i\in[G]}c_\star \,\epsilon_{c,i}\le\log(\sqrt{\omega_1})$; this fixes the total sequence length $L\coloneq\sum_{i\in[G]}L_i$; 
$ii)$ the condition $\lambda\le\log(\sqrt{\omega_1\omega_2})$ then becomes $c_\star \,L\,\epsilon_\mathtt{Q}\le \log(\sqrt{\omega_2})$, which is true by assumption. 
\end{proof}

\begin{lemma}[Quantum resource requirements for solving Problem \ref{problem} using QEC only]\label{lemma:maxG_ECstrategy}
In the same setup of Thm.~\ref{thm:main}, let $\eta_1,\eta_2\in(0,1)$ be the fractions of the total estimation error $\varepsilon$ due to compilation and statistical errors, respectively. Let $L_i(\varepsilon)$ be the length of the sequence that compiles the gate $U_i$ to diamond norm precision $\epsilon_{c,i}\le\eta_1\,\varepsilon/G$, and $L(\varepsilon)\coloneq\sum_{i\in[G]}L_i(\varepsilon)$. 
Problem \ref{problem} is solvable by \emph{QEC only} using $\frac{1}{2}\eta_2^{-2}\|O\|^{2}\log(2/\delta)/\varepsilon^{2}$ samples if and only if 
\begin{align}\label{eq:feasibility_QEC}
    L(\varepsilon)\,\epsilon_{\mathtt{Q}}\leq(1-\eta_1-\eta_2)\,\varepsilon \,.
\end{align}
where the righthand side is the error due strictly to logical noise. 
\end{lemma}

\begin{proof}[Proof of Lemma \ref{lemma:maxG_ECstrategy}]
Repeated use of Lemma \ref{lemma:noisySK} followed by application of the inequality $\|\mathcal{C}_1\mathcal{C}_0-\mathcal{D}_1\mathcal{D}_0\|_\diamond \le \|\mathcal{C}_0-\mathcal{D}_0\|_\diamond+\|\mathcal{C}_1-\mathcal{D}_1\|_\diamond$ allow us to conclude that
$$\Big|\trace(O\,\mathcal{U}(\rho))-\trace(O\,\widetilde{\mathcal{U}}_{c}(\rho))\Big|\leq \sum_{i\in[G]}(L_i\epsilon_{\mathtt{Q}}+\epsilon_{c,i})\,,$$
where $L_i=L_i(\varepsilon)$.
Therefore, if $y$ is the random variable representing output of measuring $O$ on the noisy state $\widetilde{\mathcal{U}}_{c}(\rho)$, and $\overline{y}$ is its empirical mean, the total estimation error can written as
$$\Big|\overline{y}-\trace(O\,\mathcal{U}(\rho))\Big| \leq \Big|\overline{y}-\trace(O\,\mathcal{U}_c(\rho))\Big|+\sum_{i\in[G]}(L_i\epsilon_{\mathtt{Q}}+\epsilon_{c,i}),$$
or in other words, $\overline{y}$ is an unbiased estimator for $\trace(O\,\mathcal{U}_c(\rho))$, but a biased estimator for $\trace(O\,\mathcal{U}(\rho))$ with bias given by $\sum_{i\in[G]}(L_i\epsilon_{\mathtt{Q}}+\epsilon_{c,i})$. Now let $\epsilon_{s}$ be an upper bound for the estimation error of $\trace(O\, 
\mathcal{U}_c(\rho))$.
Then, in order to have $|\overline{y}-\trace(O\,\mathcal{U}(\rho))|\leq\varepsilon$ we need the following conditions:
$i)$ for every $i\in[G]$, set the compilation error to $\epsilon_{c,i}\le\eta_1\,\varepsilon/G$ -- this fixes the total sequence length $L\coloneq\sum_{i\in[G]}L_i$ and ensures a total compilation error $\epsilon_c\le\eta_1\,\varepsilon$; 
$ii)$ collect enough statistics so as to make the statistical error $\epsilon_{s}=\eta_2\,\varepsilon$ -- using Hoeffding's inequality, this can be achieved with success probability $1-\delta$ using 
\begin{align}\label{eq:samples-qec}
    S_{\text{QEC}}\coloneq\eta_2^{-2}\|O\|^2\log(2/\delta)/(2\varepsilon^2)
\end{align} 
samples; 
$iii)$ then, imposing $(\eta_1+\eta_2)\,\varepsilon+L\,\epsilon_{\mathtt{Q}}\leq\varepsilon$ we obtain the remaining condition given by the lemma, which concludes the proof.
\end{proof}

\begin{corollary}
\label{cor:unsolvable}
Under the conditions of Lemma \ref{lemma:maxG_ECstrategy}, if $L\,\epsilon_{\mathtt{Q}}>(1-\eta_1)\varepsilon$ then correctness of the QEC-only strategy in solving Problem \ref{problem} cannot be guaranteed even in the limit of infinitely many samples. Also, given a fixed precision and a fixed amount of samples, there must be a maximal $\eta_1$ (and therefore a maximal $L$) for which Problem \ref{problem} can be guaranteed to be solvable in this device.
\end{corollary}

\begin{proof}
The first conclusion follows directly from Lemma \ref{lemma:maxG_ECstrategy} by taking the limit $\eta_2\rightarrow0$, while the second follows from Lemma \ref{lemma:maxG_ECstrategy} by the fact that $L$ grows as $\eta_1$ grows regardless of the compilation algorithm, which implies that the left-hand side of the feasibility condition grows without bounds, while the right-hand side is eventually constant.
\end{proof}

The following Lemma shows that another strategy for solving Problem 1 based on standard PEC after compiling the circuit into ideal gates from the universal set $V$ cannot succeed for arbitrary $\varepsilon$. 
\begin{lemma}[Solving Problem 1 with standard PEC]\label{lemma:standardPEC}
In the same setup of Thm.~\ref{thm:main}, let $\eta\in(0,1)$ and $\xi>1$ be constants, $L_i(\varepsilon)$ be the length of the sequence that compiles the gate $U_i$ to diamond norm precision $\epsilon_{c,i}\le\eta\,\varepsilon/G$, and $L(\varepsilon)\coloneq\sum_{i\in[G]}L_i(\varepsilon)$. 
Problem \ref{problem} is solvable using \emph{standard PEC} with $\frac{1}{2}\xi^{2}\|O\|^{2}\log(2/\delta)/\varepsilon^{2}$ samples if and only if 
\begin{align}\label{eq:feasibility_PEC}
    L(\varepsilon)\,\epsilon_\mathtt{Q} \le \frac{\log((1-\eta)\xi)}{c_\star}\,. 
\end{align}
\end{lemma}

\begin{proof}
We start by compiling each gate $U_i$ in the original circuit description $U=\prod_{i\in[G]} U_i$ into a sequence of universal gates $V_{c,i}=\prod_{\ell\in[L_i(\varepsilon)]} V_{i_\ell}$ satisfying $\norm{\mathcal{U}_i-\mathcal{V}_{c,i}}_\diamond\le\epsilon_{c,i}$, such that the full compiled circuit $U_c\coloneq\prod_{i\in[G]} V_{c,i}$ satisfies $\norm{\mathcal{U}-\mathcal{U}_{c}}_\diamond\le\epsilon_\text{c}\coloneq\sum_{i\in[G]}\epsilon_{c,i}$. 
We then compute a statistical estimate $y$ of $\Tr(OU_c^\dagger\rho U_c)$, which is a biased estimator of $\Tr(OU^\dagger\rho U)$ with bias $\epsilon_\text{c}$. The total estimation error is therefore $\big\vert y-\Tr(OU^\dagger\rho U)\big\vert\le\epsilon_\text{c}+\epsilon_\text{s}$, where $\epsilon_\text{s}$ is the statistical error. This can be made to agree with the target precision $\varepsilon$ by allocating a fraction of $\varepsilon$ to compilation and the remaining to statistics, \emph{i.e.}, $\epsilon_\text{c}=\eta\,\varepsilon$ and $\epsilon_\text{s}=\varepsilon-\epsilon_\text{c}=(1-\eta)\,\varepsilon$ with $\eta\in(0,1)$. 
The former can be ensured by compiling each gate with precision $\epsilon_{c,i}\le\eta\,\varepsilon/G$. 
The estimation of $y$ to precision $\epsilon_\text{s}$ is done by applying standard PEC independently to each ideal channel $\mathcal{V}_{i_\ell}$ as $\mathcal{V}_{i_\ell}=\sum_{j\in[\vert\widetilde{B}\vert]}a_j\widetilde{B}_j$. The negativity of each term is trivially controlled since, by assumption, $\widetilde{\mathcal{V}}_{i_\ell}$ (the noisy version of the gate $\mathcal{V}_{i_\ell}$ we are trying to expand) is an implementable operation and therefore should be included in the basis: namely, $\gamma_{i_\ell}=1+\|\mathcal{V}_{i_\ell}-\widetilde{\mathcal{V}}_{i_\ell}\|_{\widetilde{B}}\le1+c_\star\,\epsilon_\mathtt{Q}$
since $\|\mathcal{V}_{i_\ell}-\widetilde{\mathcal{V}}_{i_\ell}\|_\diamond\le\epsilon_\mathtt{Q}$ by Eq.~\eqref{eq:noisychannels}. 
The total negativity then is bounded as $\gamma\le \prod_{i\in[G]}\prod_{\ell\in[L_i(\varepsilon)]}(1+c_\star\epsilon_\mathtt{Q})=(1+c_\star\epsilon_\mathtt{Q})^{L(\varepsilon)}\le e^{c_\star L(\varepsilon)\epsilon_\mathtt{Q}}$, where $L(\varepsilon)=\sum_{i\in[G]}L_i(\varepsilon)$. 
The number of required samples is $\frac{1}{2}\gamma^2\norm{O}^2\log(2/\delta)/\epsilon_\text{s}^2=\frac{1}{2}\gamma^2(1-\eta)^{-2}\norm{O}^2\log(2/\delta)/\epsilon^2$. The overhead factor $\gamma^2(1-\eta)^{-2}\le e^{2c_\star L\epsilon_\mathtt{Q}}(1-\eta)^{-2}$ can be controlled to be $\le\xi^2$ as long as $L(\varepsilon)\le\log((1-\eta)^{2}\xi^2)/(2c_\star\epsilon_\mathtt{Q})$. 
\end{proof}

\section{Stability analysis} \label{app:stability}

Here we analyze the stability of CIPEC under approximate noise characterization. 
Recall that the implementable operations in Eq. \eqref{eq:noisychannels} have the form $\widetilde{\mathcal{A}}_j=\mathcal{E}_j\circ\mathcal{A}_j$ where $\mathcal{E}_j$ are the noise channels affecting the device. 
It is important to notice that we will never be able to perfectly characterize each $\mathcal{E}_j$ -- the best one can hope for are estimates ${\mathcal{E}}^\prime_j$ obtained via gate-set tomography \cite{merkel2013,nielsen2021}. 
Therefore, we must ensure that our method is stable against characterization errors. 
The following theorem proves that solving Problem \ref{problem} with a target precision $\varepsilon$ using CIPEC requires characterizing the error channels to diamond-norm precision $\varepsilon/(2L\norm{O})$ and running CIPEC with a target error $\varepsilon/2$.

\begin{reptheorem}{thm:stability}[Stability of CIPEC]
Let $\langle O\rangle$ be the exact expectation value,  $\overline{O}^{\,\prime},\overline{O}$ be the outputs of Alg. 5 using, respectively, the characterized noise channels $\widetilde{\mathcal{A}}^\prime_j=\mathcal{E}^\prime_j\circ\mathcal{A}_j$ and the true noisy channels $\widetilde{\mathcal{A}}_j=\mathcal{E}_j\circ\mathcal{A}_j$, and $\epsilon_\text{char}\coloneqq \max_{j\in[\vert\mathcal{A}\vert]}\norm{\mathcal{E}^\prime_j-{\mathcal{E}}_j}_{\diamond}$ be the maximum characterization error. 
If $\epsilon_\text{char}\leq\frac{\varepsilon}{L\norm{O}}$ then $\big\vert \overline{O}^{\,\prime}-\langle O\rangle\big\vert\leq 2\varepsilon$. 
\end{reptheorem}

\begin{proof}
Here for simplicity we denote by $\alpha_i=\mathbf{b}_i\cup\{1\}$ the extended vector of coefficients in Eq.~\eqref{eq:lp} that includes also the unit coefficient of the compilation term. 
Let $p(\alpha)=\prod_{i=1}^Gp_i(\alpha_i)$ be the distribution defined by the coefficients in Alg.~\ref{alg:CIPEC} \emph{with respect to the perfectly characterized} $\mathcal{E}_j$, and $\prod_{i\in[G]}\widetilde{\mathcal{B}}^i_{\alpha_i}$ be the corresponding sampled circuit, with $\widetilde{\mathcal{B}}^i_{\alpha_i}$ denoting the element of the augmented basis $\widetilde{B}^i$ with coefficient $\alpha_i$ in the decomposition of $U_i$. 
Similarly, denote by $\prod_{i\in[G]}\widetilde{\mathcal{B}}^{\prime\,i}_{\alpha_i}$ the corresponding circuit obtained by replacing each $\mathcal{E}_j$ by the imperfectly characterized noise channel $\mathcal{E}^\prime_j$. 
It follows from repeated use of the diamond norm inequality $\|\mathcal{C}_1\mathcal{C}_0-\mathcal{D}_1\mathcal{D}_0\|_\diamond \le \|\mathcal{C}_0-\mathcal{D}_0\|_\diamond+\|\mathcal{C}_1-\mathcal{D}_1\|_\diamond$ that
\begin{align*}
\norm{\prod_{i\in[G]}\widetilde{\mathcal{B}}^i_{\alpha_i}-\prod_{i\in[G]}\widetilde{\mathcal{B}}^{\prime\,i}_{\alpha_i}}_\diamond 
&\le \sum_{i\in[G]} \norm{\widetilde{\mathcal{B}}^i_{\alpha_i}-\widetilde{\mathcal{B}}^{\prime\,i}_{\alpha_i}}_{\diamond} \\
&\le \sum_{i\in[G]} L_i\max_{j\in[\vert\mathcal{A}\vert]}\norm{\mathcal{E}^\prime_j-{\mathcal{E}}_j}_{\diamond}\\
&= L\,\epsilon_\text{char},
\end{align*}
where the second inequality uses the fact that each $\widetilde{\mathcal{B}}^i_{\alpha_i}$ contains a sequence of either $L_i$ or $D<L_i$ implementable operations $\widetilde{\mathcal{A}}_j=\mathcal{E}_j\circ\mathcal{A}_j$ and $\norm{\mathcal{A}_j}_\diamond=1$. 
As a result,
\begin{align*}
\norm{\sum_{\alpha}p(\alpha)\left(\prod_{i\in[G]}\widetilde{\mathcal{B}}^i_{\alpha_i}-\prod_{i\in[G]}\widetilde{\mathcal{B}}^{\prime\,i}_{\alpha_i}\right)}_\diamond\le L\,\epsilon_\text{char}\,.
\end{align*}
Given that each statistical run outputs a random variable with absolute value at most $\|O\|$, it follows that $\big\vert\overline{O}^{\,\prime}-\overline{O}\big\vert\le L\|O\|\epsilon_\text{char}$. 
Since $\left\vert \overline{O}-\langle O\rangle\right\vert \leq \varepsilon$, the claim follows immediately by triangle inequality,
\begin{align}
\big\vert \overline{O}^{\,\prime}-\langle O\rangle\big\vert\leq \varepsilon+L\|O\|\epsilon_\text{char}\,.    
\end{align}
In particular, $\big\vert \overline{O}^{\,\prime}-\langle O\rangle\big\vert\leq 2\varepsilon$ for $\epsilon_\text{char}\leq\frac{\varepsilon}{L\norm{O}}$. 
\end{proof}
Thus, we see that the protocol is stable with respect to imperfect characterization. Furthermore, if the (logical) noise in the device does not change over time, we only need to perform this characterization once to implement all circuits having a number of gates $G$ and target precision $\varepsilon$ of the same order.

\begin{repcorollary}{cor:cipec_noiseless}[CIPEC with a noiseless basis]
If $L\,\epsilon_\mathtt{Q}\le\varepsilon/\|O\|$, CIPEC using a noiseless basis $B$ estimates the observable $\langle O\rangle$ with precision $2\varepsilon$. 
\end{repcorollary}
\begin{proof}
Let $I$ denote the two-qubit identity channel. It follows that
\begin{align}
\epsilon_\text{char}&\coloneqq\max_{j\in\big[\vert A\vert\big]}\|\mathcal{E}_j-\mathcal{E}^\prime_j\|_\diamond\notag\\
&\le \max_{j\in\big[\vert A\vert\big]}\|\mathcal{E}_j-I\|_\diamond + \max_{j\in\big[\vert A\vert\big]}\|I-\mathcal{E}^\prime_j\|_\diamond \notag\\
&= \epsilon_\mathtt{Q} + \max_{j\in\big[\vert A\vert\big]}\|I-\mathcal{E}^\prime_j\|_\diamond.    
\end{align}
Therefore, if resources allow a device with worst-case noise  $\epsilon_\mathtt{Q}\le\varepsilon/(L\|O\|)$, the stability condition $\epsilon_\text{char}\le\varepsilon/(L\|O\|)$ of Thm. \ref{thm:stability} is trivially ensured by characterizing the noise channels as identity channels ($\mathcal{E}^\prime_j=I$ for all $j\in[\vert\mathcal{A}\vert]$). In other words, the quasi-probability decomposition in Eq. (5) can be done entirely with respect to noise-free gates and basis elements. 
\end{proof}

\section{Estimates for $c_1$ and $c_2$}\label{app:c1c2}

\begin{figure*}[!ht]
    \includegraphics[width=\textwidth]{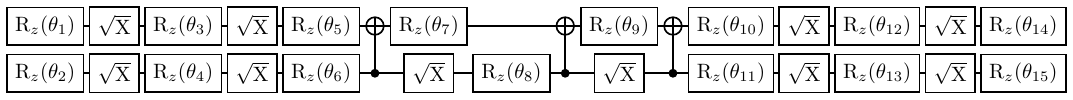}
    \caption{{\bf (Two-qubit gate decomposition used in our Clifford+T compilation)} Decomposition of an arbitrary two-qubit unitary in terms of $15\,R_z(\theta$), $10\,\sqrt{X}$, and $3\,\mathrm{CNOT}$ gates. Each $\sqrt{X}$ is exactly compiled into Clifford+T as $\sqrt{\mathrm{X}}=e^{i\pi/4}\mathrm{S}^\dagger\mathrm{H}\mathrm{S}^\dagger$, while $R_z$ gates are approximately compiled to any target precision using {\sc{gridsynth}} \cite{ross2016}.}
    \label{fig:2q-rz-sx}
\end{figure*}

Here we study the sequence length $L$ required for the Clifford+T compilation of Haar random two-qubit unitaries to a prescribed diamond norm precision $\epsilon_c$. 
Based on the Solovay-Kitaev theorem \cite{dawson2005solovaykitaevalgorithm}, we assume the following power-log scaling for the average case sequence length
\begin{align}\label{eq:skscaling}
 L = c_1 \log^{c_2}(1/\epsilon_c)\,.
\end{align}
The goal is to obtain numerical estimates for the constants $c_1$ and $c_2$. 
To the best of our knowledge, there is no straightforward algorithm to synthesize arbitrary two-qubit unitaries using the Clifford+T gateset with an end-to-end scaling analysis in diamond norm error. 
For that reason, we adopt a heuristic strategy described below. 
We stress that the values obtained for $c_1,c_2$ are heuristic estimates suitable only for our chosen compilation strategy, though the scaling should hold generally.

We draw random unitaries $U_\text{Haar}$ from the Haar distribution and, for each of them, proceed as follows: 
$i)$ decompose it in terms of $15\,R_z$, $10\,\sqrt{X}$, and $3\,\mathrm{CNOT}$ gates as illustrated in Fig. \fig{2q-rz-sx} using the {\sc qiskit} transpiler \cite{qiskit2024}; 
$ii)$ sample a target error upper bound $\epsilon^\prime_c$ from the uniform distribution in the interval $\big[\log\log10^{-2},\log\log10^{-9}\big]$;
$iii)$ synthesize each $\mathrm{R}_z$ into Clifford+T gates to precision $\epsilon^\prime_c/15$ using the {\sc gridsynth} algorithm \cite{ross2016} and each $\sqrt{\mathrm{X}}$ by its exact Clifford decomposition $\sqrt{\mathrm{X}}=e^{i\pi/4}\mathrm{S}^\dagger\mathrm{H}\mathrm{S}^\dagger$;
$iv)$ reduce the resulting word using trivial Clifford+T identities (\emph{e.g.}, $\mathrm{S}^\dagger\mathrm{S}=\mathrm{S}\mathrm{S}^\dagger=\openone$) and record the word length $L$ as well as the resulting diamond norm compilation error $\epsilon_c\le\epsilon^\prime_c$ returned by {\sc gridsynth}. 
Since {\sc gridsynth} is non-deterministic, we repeat steps $iii)$ and $iv)$ above $50$ times and keep only the pair $(L,\epsilon_c)$ corresponding to the smallest $L$ (i.e., best-case synthesis of $U_\text{Haar}$). 
Fig. \fig{c1c2scaling} shows a scatter plot of $L$ vs. $\log(1/\epsilon_c)$ (in log-log scale) for $10^4$ Haar random unitaries, from which the constants $c_1$ and $c_2$ (for the average case) can be estimated by fitting a linear model, from which we obtain $c_1=210.36$ and $c_2=0.75$.

\begin{figure}[!h]
    \includegraphics[width=\columnwidth]{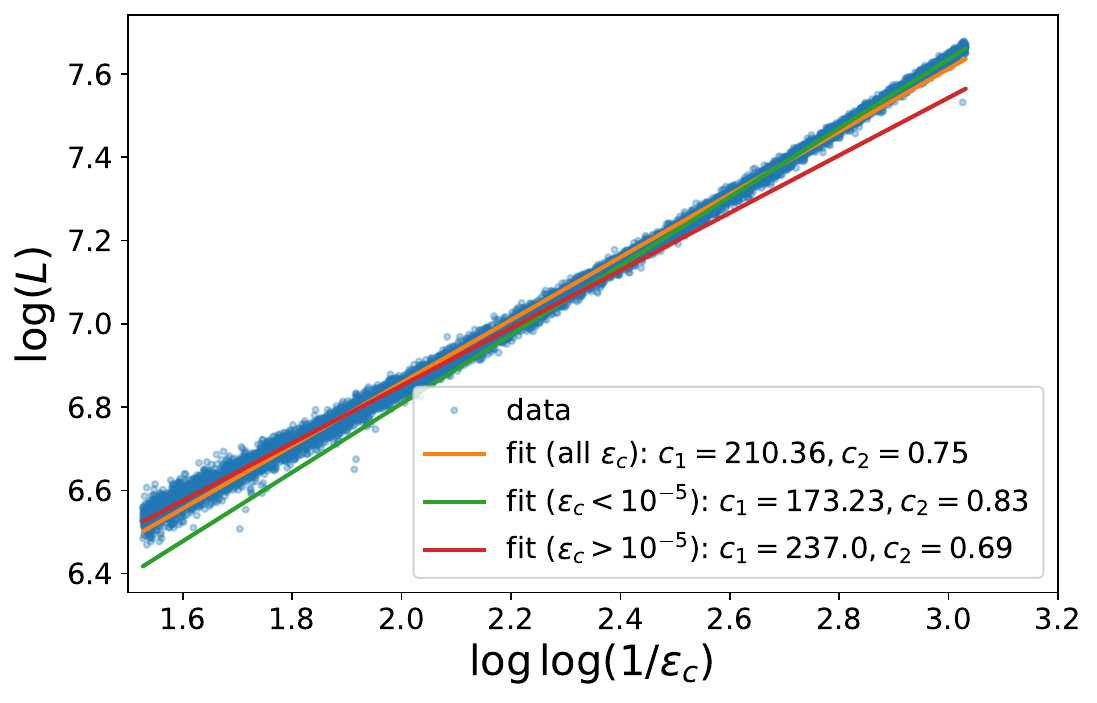}
    \caption{{\bf (Clifford+T compilation scaling)}. The points show the total number of gates $L$ (in $\log$-scale) \emph{vs.} the inverse compilation error $\epsilon_c$ (in $\log\log$-scale) for $10^4$ Haar random 2-qubit unitaries $U_\text{Haar}$. Each $U_\text{Haar}$ is first decomposed as in Fig. \ref{fig:2q-rz-sx} and then synthesized into Clifford+T gates using {\sc gridsynth} with a target compilation error uniformly sampled (in log-log scale) from $[10^{-9},10^{-2}]$. Assuming the Solovay-Kitaev-like polylogarithmic scaling \eqref{eq:skscaling} in the average case, we determine the constants $c_1,c_2$ through a linear fit (orange curve). 
    For comparison, we show how these values change when the fit is restricted to the asymptotic regime (green curve) or to the low precision regime (red curve).}
    \label{fig:c1c2scaling}
\end{figure}

\end{document}